\documentclass[a4paper, 11pt]{article}
 
\usepackage{microtype}
\usepackage{fullpage}
\usepackage{amstext,amsmath,amsthm,amssymb,amscd}
\usepackage{paralist}
\usepackage{algorithm}
\usepackage{algpseudocode}
\usepackage{xspace}
\usepackage[UKenglish]{babel}
\usepackage{authblk}
\usepackage{nccmath}
\usepackage{bm}
\usepackage{hyperref}
\usepackage{float}

\usepackage{tikz}
\usetikzlibrary{shapes,snakes}

\DeclareRobustCommand\mytikzmark{\raisebox{2pt}{\tikz \draw[very thick,dashed, red] (0,0) -- (5mm,0);} }

\title{A randomized polynomial kernel for Subset Feedback Vertex Set}
\author{Eva-Maria C. Hols}
\author{Stefan Kratsch}
\affil{University of Bonn, Germany, \{hols,kratsch\}@cs.uni-bonn.de}

\theoremstyle{plain}

\newtheorem{claim}{Claim}
\newtheorem{theorem}{Theorem}
\newtheorem{definition}{Definition}
\newtheorem{lemma}{Lemma}
\theoremstyle{remark}

\newcommand{\prob}[1]{\textsc{\lowercase{#1}}}
\newcommand{\problem}[1]{\prob{#1}\xspace}

\newcommand{\Oh}{\mathcal{O}}
\newcommand{\cP}{\mathcal{P}}

\newcommand{\ESFVS}{\prob{Edge Subset FVS}\xspace}
\newcommand{\PESFVS}{\prob{Pair-Constrained Edge Subset FVS}\xspace}
\newcommand{\SFVS}{\prob{Subset FVS}\xspace}
\newcommand{\FVS}{\prob{FVS}\xspace}
\newcommand{\DFVS}{\prob{DFVS}\xspace}
\newcommand{\DTMWC}{\prob{Deletable Terminal Multiway Cut}\xspace}
\newcommand{\MWC}{\prob{Multiway Cut}\xspace}
\newcommand{\torso}{\mathrm{torso}}

\newcommand{\problembox}[4]{
\begin{center}
\framebox{
    \begin{minipage}{0.95\textwidth}
    \problem{#1}\hfill \textbf{Parameter:} #2 \\
    \textbf{Input:} #3 \\
    \textbf{Question:} #4
\end{minipage}
}
\end{center}
}

\newcounter{rule}

\makeatletter
\def\namedlabel#1#2{\begingroup
    #2%
    \def\@currentlabel{#2}%
    \phantomsection\label{#1}\endgroup
}
\makeatother

\begin{document}

\maketitle

\begin{abstract}
The \problem{Subset Feedback Vertex Set} problem generalizes the classical \problem{Feedback Vertex Set} problem and asks, for a given undirected graph $G=(V,E)$, a set $S \subseteq V$, and an integer $k$, whether there exists a set $X$ of at most $k$ vertices such that no cycle in $G-X$ contains a vertex of $S$. It was independently shown by Cygan et al.\ (ICALP '11, SIDMA '13) and Kawarabayashi and Kobayashi (JCTB '12) that \problem{Subset Feedback Vertex Set} is fixed-parameter tractable for parameter $k$. Cygan et al.\ asked whether the problem also admits a polynomial kernelization.

We answer the question of Cygan et al.\ positively by giving a randomized polynomial kernelization for the equivalent version where $S$ is a set of edges. . 
In a first step we show that \problem{Edge Subset Feedback Vertex Set} has a randomized polynomial kernel parameterized by $|S|+k$ with $\Oh(|S|^2k)$ vertices. For this we use the matroid-based tools of Kratsch and Wahlstr\"om (FOCS '12) that for example were used to obtain a polynomial kernel for \problem{$s$-Multiway Cut}. Next we present a preprocessing that reduces the given instance $(G,S,k)$ to an equivalent instance $(G',S',k')$ where the size of $S'$ is bounded by $\Oh(k^4)$. 
These two results lead to a polynomial kernel for \problem{Subset Feedback Vertex Set} with $\Oh(k^9)$ vertices.
\end{abstract}

%%%%%%%%%%%%%%%%%%%%%%%%%%%%%%%%%%%%%%%%%%%%%%%%%%%%%%%%%%%%%%%%%%%%%%%%%%%%%%%%%%%%%%%%%%%%%%%%%%%%%%%%%%%%%%%%%%%%%%%%%%%%%%%%%%%%%%%%%%%%%%%%%%%%%%%%%%%%%%%%%%%%%%%%%%%%%%%%%%%%%%%%%%%%%%%%%%%%%%%%%%%%%%%%%%%%%%%%%%%%%%%%%%%%%%%%%%%%%%%%%%%%%%%%%%%%%%%%%%%%%%%%%%%%%%%%%%%%%%%%%%%%%%%%

\section{Introduction}

In the \problem{Subset Feedback Vertex Set} (\SFVS) problem we are given an undirected graph $G=(V,E)$, a set of vertices $S\subseteq V$, and an integer $k$, and have to determine whether there is a set $X$ of at most $k$ vertices that intersects all cycles that contain at least one vertex of $S$. Clearly, because we can choose $S=V$, this is a generalization of the well-studied \problem{Feedback Vertex Set} (\FVS) problem where, given $G$ and $k$, we have to determine whether some set $X$ of at most $k$ vertices intersects \emph{all cycles} in $G$.
\problem{Feedback Vertex Set} has been extensively studied in parameterized complexity: It is known to be fixed-parameter tractable (FPT) with parameter $k$, i.e., solvable in time $f(k)\cdot |V|^c$, and after a series of improvements the fastest known algorithms take deterministic time $\Oh^*(3.619^k)$ \cite{KociumakaP14} and randomized time $\Oh^*(3^k)$ \cite{CyganNPPRW11}. It is also known to admit a \emph{polynomial kernelization}~\cite{BurrageEFLMR06}, i.e., there is an efficient algorithm that reduces any instance $(G,k)$ of \FVS to an equivalent instance of size polynomial in $k$; the best known kernelization creates an equivalent instance with $\Oh(k^2)$ vertices~\cite{DBLP:journals/talg/Thomasse10}.

In 2011, Cygan et al.~\cite{DBLP:journals/siamdm/CyganPPW13} and Kawarabayashi and Kobayashi~\cite{DBLP:journals/jct/KawarabayashiK12} independently showed that \SFVS is FPT.
The algorithm of Cygan et al.\ runs in time $2^{\Oh(k \log k)} n^{\Oh(1)}$, while the one of Kawarabayashi and Kobayashi runs in time $\Oh(f(k) \cdot n^2m)$. Wahlstr\"om~\cite{DBLP:conf/soda/Wahlstrom14} then gave the first single-exponential algorithm with running time $4^k \cdot n^{\Oh(1)}$; an algorithm with subexponential dependence on $k$ is ruled out under the Exponential-Time Hypothesis (e.g., because \SFVS generalizes \problem{Vertex Cover}). More recently, Lokshtanov et al.~\cite{DBLP:conf/icalp/LokshtanovRS15} gave algorithms with deterministic time $2^{\Oh(k \log k)} \cdot (n+m)$ and randomized time $\Oh(25.6^k \cdot (n+m))$. 

Cygan et al.~\cite{DBLP:journals/siamdm/CyganPPW13} ask whether the \SFVS problem also admits a polynomial kernelization and suggest that the matroid-based tools of Kratsch and Wahlstr\"om~\cite{DBLP:conf/focs/KratschW12} could be applicable. The latter work uses representative sets of independent sets in matroids to obtain, amongst others, polynomial kernels for \problem{$s$-Multiway Cut} and \DTMWC (\problem{DTMWC}) with $\Oh(k^{s+1})$ and $\Oh(k^3)$ vertices, respectively. In \problem{Multiway Cut} we are given a graph $G=(V,E)$, a set $T\subseteq V$ of terminals, and an integer $k$ and have to determine whether deletion of at most $k$ non-terminal vertices separates all terminals. In \problem{$s$-Multiway Cut} the terminal set has size at most $s$, and in \problem{DTMWC} we are also allowed to delete terminals (which is essentially the same as restricting terminals to be degree one).

Interestingly, Cygan et al.~\cite{DBLP:journals/siamdm/CyganPPW13} also provide a polynomial-time reduction from \MWC to \SFVS that does not change the parameter value and, hence, is known to imply that \SFVS is at least as hard as \MWC regarding existence of polynomial kernels.
Accordingly, \MWC would be the natural next target problem for attempting to find a polynomial kernelization (after \problem{$s$-Multiway Cut} and \DTMWC).
It appears, however, that the reduction of Cygan et al.\ is from \DTMWC rather than from the more general \MWC, and it is not obvious whether similar ideas could yield a reduction from \MWC to \SFVS.

\subparagraph{Our work.}
We apply the matroid-based tools of Kratsch and Wahlstr\"om~\cite{DBLP:conf/focs/KratschW12} and develop a randomized polynomial kernelization that reduces instances $(G,S,k)$ of \SFVS to equivalent instances with at most $\Oh(k^9)$ vertices; this is our main result. Similarly to Cygan et al.~\cite{DBLP:journals/siamdm/CyganPPW13} we also work on \ESFVS where $S$ is a set of edges of $G$ and $X$ needs to intersect all cycles that contain at least one edge of $S$; \ESFVS and \SFVS are equivalent~\cite{DBLP:journals/siamdm/CyganPPW13}. The result is obtained in two parts.

In the first part (Section \ref{section:kernelization}) we establish a randomized polynomial kernelization for \ESFVS parameterized by $|S|+k$ that reduces to equivalent instances with at most $\Oh(|S|^2k)$ vertices. Note that nontrivial instances have $k<|S|$ since one could otherwise remove $S$ by deleting one endpoint of each edge in $S$. Thus, parameterization by $|S|$ suffices, but $\Oh(|S|^2k)$ gives a tighter overall bound than $\Oh(|S|^3)$.

At high level, this part is similar to the polynomial kernelization for \DTMWC. We show that certain solutions $X$, later called \emph{dominant} solutions, allow particular path packings in the underlying graph $G$. For \DTMWC this is achieved by a fairly simple replacement argument for solutions $X$ that are not sufficiently well connected to connected components of $G-X$. For \ESFVS the endpoints $T=V(S)$ of edges in $S$ can be regarded as terminals, but this gives a different separation property: Solutions $X$ need not generate many connected components in $G-X$ since only $S$-cycles need to be prevented, and components may contain many vertices of $T$. Rather, in $G-X$ there must be a tree-like (or forest-like) structure with  components without $S$-edges playing the role of nodes and with edges given by $S$. Nevertheless, using the tree-like structure, a replacement argument can be found, implying that dominant solutions must create many components in $(G-X)-S$ containing vertices of $T$ and be well connected to them. This allows to set up a gammoid on $G-S$ with sources $T$ and apply, as in~\cite{DBLP:conf/focs/KratschW12}, a result of Lov\'asz~\cite{lovasz1977flats} (made algorithmic by Marx~\cite{DBLP:journals/tcs/Marx09}) on representative sets in (linear) matroids that is then guaranteed to generate a superset of $X$. Randomization is only needed to generate a matrix representation for the gammoid.

In the second part (Section \ref{section:reducing}) we give a (deterministic) polynomial-time preprocessing that, given an instance $(G,S,k)$ of \ESFVS, returns an equivalent instance $(G',S',k')$ with $k'\leq k$ and $|S'|\in\Oh(k^4)$. Together with the randomized kernelization from the first part this implies the claimed randomized kernelization to $\Oh(k^9)$ vertices.

A reduction of the number of $S$-edges is also a crucial ingredient in the FPT algorithm for \ESFVS by Cygan et al.~\cite{DBLP:journals/siamdm/CyganPPW13}. They achieve $|S|\in\Oh(k^3)$, but it is in a slightly more favorable setting: Using iterative compression, it suffices to solve the task of finding a solution $X'$ of size $k$ when given a solution $X$ of size $k+1$. (This is well known in parameterized complexity, and we prefer not to repeat it here.) Considering some unknown solution $X'$ of size $k$, one can guess the intersection $D$ of $X'$ with $X$, by trying all $\Oh(2^{k+1})$ possibilities. For the correct guess $D=X'\cap X$, the remaining problem is to find for $(G-D,S\setminus D,k-|D|)$ a solution $Z'$ of size at most $k-|D|$ that is disjoint from $Z=X\setminus D$, since $Z'=X'\setminus D$ would be such a solution; here $S \setminus D$ denotes the set of edges in $S$ with no endpoint in $D$. Cygan et al.\ make the nice observation that the guessing also allows to assume that there is no other solution $X'$ with an even larger intersection with $X$.

In contrast, we cannot afford to run iterative compression for a kernelization to get a starting solution of size $k+1$ and, as is common, we have to start with an approximate solution $Z$, which can be assumed to be of size at most $8k$ using an $8$-approximation algorithm of Even et al.~\cite{DBLP:journals/siamcomp/EvenNZ00}. The idea of guessing the intersection of an optimal solution with $Z$ is infeasible regarding both time and the number of created instances. Thus, while several structures like $z$-flowers or disjoint $x$,$y$-paths containing $S$-edges appear in both approaches, many things have to be handled differently. For example, having $k+2$ disjoint $x$,$y$-paths containing $S$-edges for $x,y\in Z$ implies that one of $x$ and $y$ must be in every solution of size $k$; Cygan et al.\ can stop here because the solution would not be disjoint from $Z$; we need to instead store the information about $x$ and $y$ to later detect $S$-edges that can be safely removed. 
Like Cygan et al., we also use Gallai's $A$-path Theorem but we avoid the 2-expansion lemma by using the properties of a blocking set of size at most $2k$ differently. (Such a blocking set can be found if certain flowers of order $k+1$ do not exist, using Gallai's $A$-path Theorem.)
Cygan et al. compute a blocking set $B$ of size at most $3k$ to find an $F$-flower of order $|X|$ (with $F \subseteq V$ outer-abundant; see \cite[Definition 3.4]{DBLP:journals/siamdm/CyganPPW13}) under the assumption that certain $F$-flowers of order $k+1$ do not exist and show that there exists a solution that contains $X$ (under the assumption that there exists a solution that is disjoint from $F$). We cannot assume that our solution is disjoint from $F$ and have to take another approach. Moreover, we observe that $z$-flowers can be found via matroid parity on an appropriate gammoid.\footnote{The latter is deterministic by applying a specialized matroid parity algorithm due to Tong et al.~\cite{TongLV1984}.}

\section{Preliminaries}\label{section:preliminaries}

We use standard graph notation, mostly following Diestel~\cite{Diestel2005graph}. All graphs are undirected and may contain multi-edges and loops; accordingly, they may contain cycles of length one and two (formed by loops and multi-edges, respectively.)
An edge $e \in E$ is called a \emph{bridge} if $(V,E \setminus \{e\})$ has more connected components than $G$.
For a set $X \subseteq V$, let $G[X]$ denote the subgraph of $G$ induced by $X$ and let $N_G(X)$ denote the neighborhood of $X$ in $G$, i.e., $N_G(X)=\left\{ v \in V \setminus X \mid \exists u \in X \colon \{u,v\} \in E \right\}$.
Given two sets $X,Y \subseteq V$, by $E(X,Y)$ we denote the set of edges that have one endpoint in $X$ and one endpoint in $Y$.
For a set $E' \subseteq E$ of edges let $V(E')$ be the set of vertices that are incident with at least one edge in $E'$.
For $X\subseteq V$ and $F\subseteq E$ we shorthand $G-X$ for $G[V\setminus X]$ and $G-F$ for $(V(G),E(G)\setminus F)$; if $X=\{x\}$ then we may also write $G-x$ instead of $G - \{x\}$. Note that the graph $(G-X)-F$ is the same graph as the graph $(G-F)-X$ and we will drop the parentheses.

For $A\subseteq V$ a path with endpoints in $A$ and internal vertices not in $A$ is called an \emph{$A$-path}. The following theorem about $A$-paths was already used by Cygan et al.~\cite{DBLP:journals/siamdm/CyganPPW13} for \SFVS and in the quadratic kernelization for \problem{Feedback Vertex Set} by Thomass\'e~\cite{DBLP:journals/talg/Thomasse10}.

\begin{theorem}[Gallai~\cite{gallai1961maximum}] \label{theorem::gallai}
    Let $A \subseteq V$ and $k \in \mathbb{N}$. If the maximum number of vertex-disjoint $A$-paths is 
    strictly less than $k+1$, then there exists a set $B \subseteq V$ of at most $2k$ vertices that intersect every $A$-path.
\end{theorem}

In particular it is possible to find either $(k+1)$-disjoint $A$-paths or a set $B$ that intersects all
$A$-paths in polynomial time. This follows from Schrijver's proof of Gallai's theorem~\cite{DBLP:journals/jct/Schrijver01}.

Let $(G,S,k)$ be an instance of the \ESFVS problem. We call a cycle $C$ an \emph{$S$-cycle}, if
at least one edge of $S$ is contained in $C$. Let $x$ be a vertex of $V$. A set $\{C_1, C_2, \ldots, C_t\}$ of 
$S$-cycles that contain $x$ is called an \emph{$x$-flower} of order $t$, if the sets of vertices 
$C_i \setminus \{x\}$ are pairwise disjoint. 
Note that if there exists a $x$-flower of order at least $k+1$, then the vertex $x$ must be in 
every solution for $(G,S,k)$, if one exists.
A set $B \subseteq V \setminus \{x\}$ of size $t$ is called an \emph{$x$-blocker} of size $t$, 
if each $S$-cycle through $x$ also contains at least one vertex of $B$.

\subparagraph{Parameterized complexity.} 
A \emph{parameterized problem} is a language $\mathcal{Q} \subseteq \Sigma^* \times \mathbb{N}$, where 
$\Sigma$ is any finite set. The second component of an instance $(x,k)$ is called the \emph{parameter}.
We say that a parameterized problem $\mathcal{Q}$ is \emph{fixed-parameter tractable} (FPT) if there exists a computable function $f \colon \mathbb{N} \to \mathbb{N}$ and an algorithm $A$ that on input of $(x,k)\in\mathcal{Q}\times\Sigma^*$ takes time at most $f(k)\cdot |x|^{\Oh(1)}$ and correctly decides whether $(x,k) \in \mathcal{Q}$.
A \emph{kernelization} of a parameterized problem $\mathcal{Q}$ is an algorithm $K$ that on input of
$(x,k) \in \Sigma^* \times \mathbb{N}$ takes time polynomial in $|x|+k$ and returns an equivalent instance
$(x',k') \in \Sigma^* \times k$ with $|x'|+k' \leq h(k)$, where $h$ is a computable function. The function
$h$ is called the size of the kernel. We say that $K$ is a \emph{polynomial kernelization} if $h(k) \in \Oh(k^c)$ for some 
constant $c$.
The polynomial kernelization obtained in this paper is randomized, which means that there is a small chance for the reduced instance to not be equivalent to the input. The error probability can be made exponentially small in the input size without increasing the size of the kernelization. Similarly to previous work~\cite{DBLP:conf/focs/KratschW12}, the only source for error is the need to compute a matrix representation for a particular matroid (preliminaries on matroids follow below).

\subparagraph{Matroids, gammoids, and representative sets.}

A matroid $M=(U, \mathcal{I})$ consists of a finite set $U$
and a family $\mathcal{I}$ of subsets of $U$, called \emph{independent sets}, fulfilling the following
properties:
\begin{inparaenum}[(i)]
    \item $\emptyset \in \mathcal{I}$;
    \item if $X \subseteq Y$ and $Y \in \mathcal{I}$ then also $X \in \mathcal{I}$; and
    \item if $X,Y \in \mathcal{I}$ with $|X| < |Y|$ then there exists $y \in Y \setminus X$ such that 
            $X \cup \{y\} \in \mathcal{I}$.
\end{inparaenum}
The \emph{rank} of of a matroid $M$, denoted by $r(M)$, is the size of the largest independent set of 
the matroid $M$.

Let $A$ be a matrix over an arbitrary field $F$. Let $U$ be the set of columns of $A$ and let 
$\mathcal{I}$ be the family of all sets $X \subseteq U$ of columns that are linearly independent over $F$.
Then $M=(U,\mathcal{I})$ is a matroid, called the \emph{linear matroid} or \emph{vector matroid} of $A$, and we say that $A$ \emph{represents} $M$. If $M=(U,\mathcal{I})$ is representable over some field, then it is also representable by an $r(M)\times |U|$ matrix; by Gaussian elimination we can always reduce a representing matrix for $M$ to one with $r(M)$ many rows (cf.~\cite{DBLP:journals/tcs/Marx09}).
Let $M_1=(U_1,\mathcal{I}_1)$ and $M_2=(U_2,\mathcal{I}_2)$ be two matroids with $U_1 \cap U_2 = \emptyset$.
The \emph{direct sum} $M_1 \oplus M_2$ is a matroid over $U=U_1 \cup U_2$ with independent sets
$\mathcal{I} = \{X \subseteq U \mid X \cap U_1 \in \mathcal{I}_1, X \cap U_2 \in \mathcal{I}_2 \}$.
If $A_1$ and $A_2$ represent the two matroids over the same field $F$, then matrix 
$A=\mathrm{diag}(A_1,A_2)$ represents $M_1 \oplus M_2$.

Let $G=(V,E)$ be a graph that may have both directed and undirected edges and let $S \subseteq V$.  
A set $T \subseteq V$ is \emph{linked} to $S$ if there exist $|T|$ vertex-disjoint paths from $S$ to $T$.
Thus every vertex in $T$ is endpoint of a different path from $S$.
It holds that $M=(U, \mathcal{I})$, where $U \subseteq V$ and $\mathcal{I}$ contains all sets $T \subseteq U$ that are linked to $S$ in $G$, is a matroid~\cite{Perfect1968}. The matroid $M$ is also called the \emph{gammoid} on $G$ with sources $S$ and ground set $U$; if $U=V$ then $M$ is also called a \emph{strict gammoid}.
Marx~\cite{DBLP:journals/tcs/Marx09} gave a randomized polynomial-time procedure for finding a 
matrix representation of a strict gammoid. The error probability can be made exponentially small in the size of the graph. (This is the only source of randomness and error in our kernelization.) A matrix representation for a gammoid for graph $G=(V,E)$ with ground set $U\subsetneq V$ and sources $S$ can be obtained from one for the strict gammoid for $G$ and $S$ by simply deleting columns corresponding to elements of $V\setminus U$.

Let $A, B$ be independent sets in a matroid. We say that $A$ \emph{extends} $B$ if $A \cap B = \emptyset$ and 
$A \cup B$ is again an independent set. Note that from the independence of $A \cup B$ follows the 
independence of $A$ and $B$ due to the second matroid property.

\begin{definition}
    Let $M=(U,\mathcal{I})$ be a matroid, let $\mathcal{A} \subseteq \mathcal{I}$, and let $q \in \mathbb{N}$. A set $\mathcal{A}' \subseteq \mathcal{A}$ is
    \emph{$q$-representative} for $\mathcal{A}$ if for every independent set $B$ of size at most $q$
    there is a set $A \in \mathcal{A}$ that extends $B$ if and only if there is also a set $A' \in \mathcal{A}'$
    that extends $B$.
\end{definition}

Observe that if $\mathcal{A}'$ is $q$-representative for $\mathcal{A}$ and there exists a set
$A \in \mathcal{A}$ that \emph{uniquely extends} some given independent set $I$ of size at most $q$, then this implies that $A \in \mathcal{A}'$.

The following theorem of Lov\'asz~\cite{lovasz1977flats} proves that for any linear matroid there exist small representative sets. It was made algorithmic by Marx~\cite{DBLP:journals/tcs/Marx09} and, thus, permits to find representative sets in polynomial time when given a matrix representation of the matroid. A faster algorithm for this task was developed recently by Fomin et al.~\cite{FominLS14}.

\begin{lemma}[Lov\'asz~\cite{lovasz1977flats}, Marx~\cite{DBLP:journals/tcs/Marx09}] \label{lemma::representative}
    Let $M$ be a linear matroid of rank $q+p$, and let $\mathcal{T} = \{I_1, I_2, \ldots, I_t\}$ be a 
    collection of independent sets, each of size $p$. If $|\mathcal{T}| > \binom{q+p}{p}$, then there is a 
    set $I \in \mathcal{T}$ such that $\mathcal{T} \setminus \{I\}$ is $q$-representative for $\mathcal{T}$.
    Furthermore, given a representation $A$ of $M$, we can find such a set $I$ in $f(q,p) \cdot (\|A\|t)^{\Oh(1)}$ time.
\end{lemma}

Given a gammoid $M$ we can compute in randomized polynomial-time a representation of the gammoid. 
Together with Theorem \ref{lemma::representative} it follows that given a gammoid $M$ and a collection 
$\mathcal{T} = \{I_1,\ldots, I_t\}$ of independent sets, each of size $p$, we can find in randomized 
polynomial time a set $\mathcal{T}' \subseteq \mathcal{T}$ of size at most $\binom{q+p}{p}$ that is 
$q$-representative for $\mathcal{T}$.

%%%%%%%%%%%%%%%%%%%%%%%%%%%%%%%%%%%%%%%%%%%%%%%%%%%%%%%%%%%%%%%%%%%%%%%%%%%%%%%%%%%%%%%%%%%%%%%%%%%%%%%%%%%%%%%%%%%%%%%%%%%%%%%%%%%%%%%%%%%%%%%%%%%%%%%%%%%%%%%%%%%%%%%%%%%%%%%%%%%%%%%%%%%%%%%%%%%%%%%%%%%%%%%%%%%%%%%%%%%%%%%%%%%%%%%%%%%%%%%%%%%%%%%%%%%%%%%%%%%%%%%%%%%%%%%%%%%%%%%%%%%%%%%%

\section{Randomized polynomial kernelization for parameter \texorpdfstring{$\boldsymbol{|S|+k}$}{|S|+k}}\label{section:kernelization}

In this section we present a randomized polynomial kernelization for \ESFVS parameterized by $|S|+k$. Because deletion of one endpoint of each edge in $S$ always constitutes a feasible solution, nontrivial instances have $|S|>k$. Thus, our kernelization also works for parameter $|S|$ alone. However, to achieve a better bound for \ESFVS parameterized by $k$ only it is beneficial to give the kernel size in terms of $|S|$ and $k$ rather than $|S|$ alone.

We use representative sets of independent sets of matroids to obtain a kernel of size $\Oh(|S|^2 k)$.
Our approach is similar to the kernelization of \problem{Deletable Terminal Multiway Cut$(k)$} 
\cite{DBLP:conf/focs/KratschW12}. As in that paper we construct path packings such that certain vertices can be shown to be in a representative set.  Note that, unlike for multiway cut-type problems, a solution $X\subseteq V$ will not necessarily create many connected components. Rather, as used also in the FPT algorithm of Cygan et al.~\cite{DBLP:journals/siamdm/CyganPPW13}, it creates a particular tree-like structure in $G-X$. Nevertheless, endpoints of edges in S, denoted $T:=V(S)$, will play the role of terminals that need to be separated in a certain way; hence a vertex $x$ in $T$  is called a \emph{terminal}. We will focus on the graph $G-S$, i.e., with edges of $S$ deleted, in which a solution $X$ creates a grouping of (not deleted) terminals into connected components. The structure of these components will be crucial for a replacement argument (Lemma \ref{lemma::replacement}) that leads to the required path packing; this constitutes one of the key arguments for our result. 

The kernelization consists of four steps. In the first step we show that if an instance is YES 
then there exists a solution $X$ with a certain path packing from $T$ to $X$. Then we define an appropriate gammoid to 
find in a next step a representative set of size $\Oh(|S|^2 k)$ which is (essentially) a superset of $X$ using 
Lemma \ref{lemma::representative}. Finally we explain how to reduce the graph $G$, using the superset
of the last step, to obtain an equivalent instance of \ESFVS.

\subparagraph{Analyzing solutions.} 
Let $(G,S,k)$ be a yes-instance of \ESFVS$(k+|S|)$. We say that a solution $X$ for $(G,S,k)$ is \emph{dominant}, if it has minimum size and contains a maximal number of vertices from $T$ among solutions of minimum size. 
The vertices in $X \cap T$ correspond to endpoints of edges in $S$ that we delete and the vertices in 
$X_0=X \setminus T$ block all $x$-$y$ paths with $\{x,y\} \in S_0=\{e \in S \mid e \cap X = \emptyset \}$, 
except the one that consists of the edge $\{x,y\}$. 
We show that $X$ is linked to $T$ in a strong sense, with vertices of $X_0$ playing a special role.

\begin{lemma} \label{lemma::path_packing}
    Let $X$ be a dominant solution for $(G,S,k)$ and $x$ any vertex in the set $X_0=X\setminus T$. 
    There exist $|X|+2$ paths from $T$ to $X$ in $G-S$ that are vertex-disjoint except for three paths ending in vertex $x$. Moreover, the paths can be chosen in such a way that each connected component of $G-X-S$ is intersected by at most one path.
\end{lemma}

We use Hall's Theorem and the lemma below to prove this. For this purpose we use the two graphs $G-X$ and $G-X-S$ which simplify the analysis of a dominant solution. We call a connected component $K$ of $G-X-S$ \emph{interesting} if it contains a terminal, i.e., if $T \cap V(K)=(T\setminus X)\cap V(K) \neq \emptyset$, and we say that $x \in X_0$ \emph{sees} an interesting component $K$ if $x$ is adjacent to a vertex of $K$ in $G$. We extend this definition by saying that $Y \subseteq X_0$ sees an interesting component $K$ if at least one vertex $y \in Y$ sees $K$.

\begin{lemma}\label{lemma::replacement}
    If $X$ is a dominant solution then every nonempty set $Y \subseteq X_0$ sees at least $|Y|+2$ 
    interesting components of $G-X-S$.
\end{lemma}

\begin{proof}
Assume for contradiction that there exists a nonempty set $Y\subseteq X_0$ that sees at most $|Y|+1$ interesting components of $G-X-S$. Let $\mathcal{C}_i$ denote the set of interesting components of $G-X-S$ seen by $Y$, and let $\mathcal{C}_o$ denote the other components seen by $Y$. We will show that there is an alternative solution $X'=(X\setminus Y)\cup Y'$ that is smaller than $X$ or that contains more vertices of $T$, contradicting the choice of $X$ as a dominant solution. To this end, let us consider the graphs $G-X$ and $G-(X\setminus Y)$ (in part repeating things that have been said earlier to get a self-contained proof).

In $G-X$ the components of $G-X-S$ may be connected by edges of $S$ and form a tree-structure with components playing the role of vertices and edges of $S$ whose endpoints are not deleted being the edges of the tree: (We say tree-structure, but a forest of components, connected by $S$-edges, is also fine.) There can be no cycles in this tree-structure because they would give rise to $S$-cycles in $G-X$. Moreover, any other set $X'$ of size at most $k$ such that $G-X'$ consists of components without $S$-edges that are connected in a tree-like manner by $S$-edges is also a valid solution. Note that non-interesting components of $G-X-S$ are isolated in $G-X$ because they do not contain vertices of $T$, i.e., no endpoints of $S$-edges, so they cannot be incident with $S$-edges in $G-X$.

In $G-(X\setminus Y)-S$ the components in $\mathcal{C}_i$ and $\mathcal{C}_o$ may form larger combined components because we do not delete the vertices in $Y$; let $\mathcal{C}'$ denote the set of these components. Crucially, because $Y\subseteq X_0=X\setminus T$, there are no additional vertices of $T$, i.e., $T\setminus X=T\setminus (X\setminus Y)$. Thus, in $G-(X\setminus Y)$ the set of $S$-edges incident with components in $\mathcal{C}'$ is the same as the $S$-edges incident with $\mathcal{C}_i$ in $G-X$; recall that no $S$-edges are incident with components in $\mathcal{C}_o$ in $G-X$. Note that, in general, $G-(X\setminus Y)$ will not have the tree-structure: In comparison to $G-X$ we are not deleting vertices of $Y$, which corresponds to merging some components in $\mathcal{C}_i\cup\mathcal{C}_o$. This may lead to components in $\mathcal{C}'$ that are incident with both endpoints of some $S$-edges (the equivalent of loops) and it may also create other (longer) cycles. We will see that deleting at most $|Y|$ edges of $S$, i.e., deleting a set $Y'$ of at most $|Y|$ endpoints of $S$-edges, will suffice to get the tree-structure, making $(X\setminus Y)\cup Y'$ a valid solution.

Consider a component $C'\in\mathcal{C'}$ of $G-(X\setminus Y)-S$ that fully contains all vertices of some components $C^1_i,\ldots,C^a_i\in\mathcal{C}_i$ and $C^1_o,\ldots,C^b_o\in\mathcal{C}_o$; additionally it may contain vertices of $Y$. (The fact that we must have full containment follows directly by comparing deletion of $X$ from $G-S$ with deletion of $X\setminus Y$ from $G-S$.)
In $G-(X\setminus Y)$ the component $C'$ may be incident with $S$-edges and thus be part of a larger component $C^+$; we want to see that deleting (one endpoint each of) at most $a-1$ $S$-edges from $C^+$ suffices to get the tree-structure.

In $G-X$ instead of component $C^+$ we may have several separate components because we additionally delete the vertices of $Y$. Since $Y$ sees only components in $\mathcal{C}_i\cup\mathcal{C}_o$ there are at most $a+b$ separate components ``created'' from $C^+$ by deleting $Y$ since these are all components contained in $C^+$ that are seen by $Y$. Recall that components in $\mathcal{C}_o$ are isolated in $G-X$ and contain no vertices of $T$ and, thus, they do not contribute any $S$-edges to $C^+$. It remains to consider the components $C^1_i,\ldots,C^a_i$ that are contained in $C^+$.

Assume first that all components $C^1_i,\ldots,C^a_i$ are part of a single connected component in $G-X$. (Recall that they are connected components of $G-X-S$ but may be connected by $S$-edges in $G-X$.) Thus, they are part of a single tree of components (connected by $S$-edges) and not deleting $Y$ corresponds to merging $a$ vertices in this tree into a single one. If the tree had $c$ components and, thus, $c-1$ $S$-edges then we obtain $c-a+1$ components that are connected by $c-1$ $S$-edges in $G-(X \setminus Y)$. (Recall that $Y$ contains no endpoints of $S$-edges.) It therefore suffices to delete $(c-1)-((c-a+1)-1)=a-1$ $S$-edges, i.e., to delete one endpoint of each of $a-1$ $S$-edges, to obtain the tree-structure. (Not \emph{any} $a-1$ edges are ok but we can keep any $c-a$ $S$-edges spanning the $c-a+1$ components and delete the $(c-1)-(c-a)=a-1$ remaining $S$-edges.)

In general, the components $C^1_i,\ldots,C^a_i$ may be part of several different connected components in $G-X$. Nevertheless, this still means that we have a cycle-free structure of components (seen as vertices) connected by $S$-edges. If overall the cycle-free structure has $c$ components then, being cycle-free, it has at most $c-1$ $S$-edges. Thus, merging yields $c-a+1$ components connected by at most $c-1$ $S$-edges and removing at most $a-1$ $S$-edges suffices.

Overall, we get that a component $C'\in\mathcal{C'}$ that fully contains $a$ interesting components from $\mathcal{C}_i$ requires at most $a-1$ vertex deletions of endpoints of $S$-edges to obtain the tree-structure. Since $Y$ sees at most $|Y|+1$ such components, the worst case is achieved by a single component $C'$ containing all $|Y|+1$ interesting components in $\mathcal{C}_i$; this still costs at most $(|Y|+1)-1=|Y|$ vertex deletions, as claimed.

Let $Y'$ contain all the endpoints of $S$-edges that we delete to get the tree-structure. We know that $|Y'|\leq|Y|$ and thus $|(X\setminus Y)\cup Y'|\leq |X|$. Moreover, by the initial considerations, we know that $X'=(X\setminus Y)\cup Y'$ is a feasible solution as $G-X'$ has the required tree-structure. If $|Y'|<|Y|$, including the case that $Y'=\emptyset$, then $|X'|<|X|$ as $Y\neq \emptyset$; this contradicts optimality of $X$ (required for being a dominant solution). If $|Y'|=|Y|$ then $Y'\neq \emptyset$ and $X'$ is an optimal solution that contains more vertices of $T\supseteq Y'$, contradicting the choice of $X$ as a dominant solution. Thus, every nonempty set $Y$ must see at least $|Y|+2$ connected components, as claimed.
\end{proof}

Now we are ready to give the proof of Lemma \ref{lemma::path_packing}. The argument relies on Hall's Theorem and is similar to the one for \problem{Deletable Terminal Multiway Cut}~\cite{DBLP:conf/focs/KratschW12}.

\begin{proof}[Proof of Lemma \ref{lemma::path_packing}]
    We know that every nonempty set $Y \subseteq X_0$ sees at least $|Y|+2$ interesting components. 
    To prove existence of the required path packing we construct a bipartite graph where one side consists of the 
    interesting components and the other side consists of the set $X_0$ and two copies $x',x''$ of 
    the vertex $x \in X_0$. We connect $v \in X_0$ with an interesting component $K$ if 
    $v$ sees $K$ and we connect $x'$ and $x''$ with the same interesting components as $x$.
    For this bipartite graph it holds that for all sets $Y \subseteq X_0 \cup \{x',x''\}$, the size of $N(Y)$
    is at least $|Y|$: This holds trivially for $Y=\emptyset$; assume there exists a nonempty set $Y \subseteq X_0 \cup \{x',x''\}$ such that $|N(Y)|<|Y|$. But then we have $|N(Y \setminus \{x',x''\})| \leq |N(Y)| < |Y| \leq |Y \setminus \{x',x''\}| + 2$, which is a contradiction to Lemma \ref{lemma::replacement}.
    
    Since Hall's condition is satisfied there exists a matching $M$ that covers 
    $X_0 \cup \{x',x''\}$. This matching gives rise to a path packing from $T$ to $X$ where exactly three paths end in $x$ and no other vertices occur in more than one path: For each $v\in X\cap T$ pick the path of length zero that consists only of $v$. For each edge $\{K,v\}$ in the matching $M$, where $v\in X_0\cup\{x',x''\}$, pick an arbitrary path from a terminal $t\in V(K)\cap (T\setminus X)$ to $v$ that uses only vertices from $V(K)\cup\{v\}$. (For $v\in\{x',x''\}$ let the path end in $x$ and use only vertices in $V(K)\cup\{x\}$.) Because $K$ is an interesting component a terminal $t\in V(K)\cap (T\setminus X)$ must exist, and because $K$ is a component of $G-X-S$ the path contains no other vertices of $X$. Similarly, the path cannot contain $S$-edges between vertices of $K$, and its final edge to $v$ cannot be in $S$ because $v\in X_0=X\setminus T$, i.e., because $v$ is not endpoint of any $S$-edge. Moreover, since each interesting component is matched to a single vertex $v\in X_0\cap\{x',x''\}$, all the paths are vertex-disjoint except for the three paths that share their endpoint $x$.
    This path packing, including the trivial paths from $X \cap T$ to $X \cap T$, contains $|X|+2$ paths from $T$ to $X$ in $G-S$ that are vertex-disjoint except for the three paths sharing endpoint $x$. By construction, there is at most one path to any vertex of $X_0$ starting in any interesting component $K$ of $G-X-S$, because the components are used according to the matching $M$. All further paths are of length zero, consisting of only a vertex in $X\cap T$ and are, thus, not contained in components of $G-X-S$.
\end{proof}

\subparagraph{Setting up the gammoid.}
The gammoid $M$ that we use is the direct sum of two gammoids $M_1$ and $M_2$. 
To construct gammoid $M_1$ we define a graph $G_1=(V_1,E_1)$ that is obtained from $G-S$ by adding two so called \emph{sink-only copies} $v'$ and $v''$ for every vertex $v \in V$. A sink-only copy of a vertex $v$ is a vertex $v'$ (or $v''$) that has a directed edge $(u,v')$ for each edge $\{u,v\}$; these were already used in previous work~\cite{DBLP:conf/focs/KratschW12}. Note that adding sink-only copies of vertices does not affect the possible path packings to other vertices since they can only be endpoints of paths; however, they are convenient to capture multiple vertex-disjoint paths that, intuitively, end in the same vertex.
The matroid $M_1$ is defined to the gammoid on $G_1$ with sources $T=V(S)$ and ground set $V_1=\{v,v',v''\mid v\in V\}$; note that the sink-only copies of vertices in $T$ are not sources of $M_1$.
The rank of matroid $M_1$ is $|T|$, because the set of all trivial paths is independent and at most $|T|$ vertices can be linked to $T$.

Matroid $M_2$ is the gammoid on the directed graph $G_2 = K_{k,n} = (S_2 \dot{\cup} \hat V,E_2)$ with sources $S_2$ and ground set $\hat V=\{\hat v\mid v\in V\}$; the edges in $E_2$ are directed from $S_2$ to $\hat V$.
In other words, gammoid $M_2$ is simply a uniform matroid and a (deterministic) matrix representation could also be obtained by using a Vandermonde matrix. The rank of $M_2$ is $k=|S_2|$ because no more than $|S_2|$ vertices can be linked to $S_2$ and every set of at most $k$ vertices of $\hat V$ is linked to $S_2$.

For the application of Lemma \ref{lemma::representative} we will use the matroid $M=M_1 \oplus M_2$, which has rank $|T|+k$. (Matroid $M$ can also be seen as a gammoid on the graph $G_1 \dot{\cup} G_2$ with appropriate sources and ground set but we prefer the explicit direct sum and the implied block-diagonal representation obtained below.)
Representations $A_1$ and $A_2$ for both $M_1$ and $M_2$ can be computed by a randomized polynomial-time algorithm with exponentially small error chance~\cite{DBLP:journals/tcs/Marx09}; hence we get a representation for $M$ by $\mathrm{diag}(A_1,A_2)$, i.e., the block-diagonal matrix with blocks $A_1$ and $A_2$. 
We may assume that $A_1$ has $|T|$ rows and $A_2$ has $k$ rows since this could be achieved by Gaussian elimination (cf.~\cite{DBLP:journals/tcs/Marx09}).

\subparagraph{Applying the representative set lemma.}
Let $\mathcal{T}:=\{\{v',v'',\hat v\}\mid v\in V\}$. For clarity, by the above notation, this means that $v',v''\in V_1$ and $\hat v\in \hat V$ for each $v\in V$. Using Lemma \ref{lemma::representative}
we will prove that we can compute in randomized polynomial time a $(|T|+k-3)$-representative subset $\mathcal{T}'$ of $\mathcal{T}$ that contains for all $x \in X_0=X \setminus T$ the set $\{x',x'', \hat x\}$, where $X$ is any dominant solution for $(G,S,k)$.
Lemma \ref{lemma::representative} guarantees $|\mathcal{T}'| \in \Oh((|T|+k)^3) =\Oh(|S|^3)$,
since we can compute a matrix representation of $M$ in randomized polynomial-time as described above.
We will see later that we can find a $(|T|+k-3)$-representative set of size $\Oh(|S|^2k)$ by a careful look at the proof of Lemma \ref{lemma::representative}, using the fact that $M$ is the 
direct sum of two gammoids and that all sets $\{v',v'',\hat v\}$ in $\mathcal{T}$ have two elements from the first and one element from the second gammoid; a similar argument for getting a smaller representative set was already used by Kratsch and Wahlstr{\"{o}}m~\cite{DBLP:conf/focs/KratschW12}.

To ensure that all sets $\{x',x'', \hat x \}$ with $x \in X_0$ are in $\mathcal{T}'$ we have to show that for each such set $\{x',x'', \hat x \}$ there exists an independent set $I$ of size at most $|T|+k-3$ such that $\{x',x'', \hat x \}$ uniquely 
extends $I$ among triplets in $\mathcal{T}$. This directly implies that $\{x',x'', \hat x \}$ must be in every $(|T|+k-3)$-representative set $\mathcal{T}'$ of $\mathcal{T}$.
\begin{lemma}
    Let $X$ be a dominant solution for $(G,S,k)$ and let $T=V(S)$.
    For all $x \in X_0=X\setminus T$ there exists an independent set $I$ of size at most $|T|+k-3$ in $M$ such that
    $\{x',x'', \hat x \}$ uniquely extends $I$.
\end{lemma}

\begin{proof}
Let $x$ be an arbitrary vertex of $X_0$.
In a first step we define an independent set $I$ and show in a second step that $\{x',x'',\hat x \}$ 
uniquely extends $I$.
Applying Lemma \ref{lemma::path_packing} implies the existence of a path packing $\mathcal{P}$ of $|X|+2$ paths from $T$ to $X$ in $G-S$ that are vertex-disjoint except for three paths ending in $x$ and such that each connected component of $G-X-S$ is intersected by at most one path of $\mathcal{P}$. This directly implies a path packing $\mathcal{P}_1$ in $G_1$ from $T$ to $X\cup\{x',x''\}$ that is (fully) vertex-disjoint. We retain the property that at most one path intersects the vertex set of any component of $G-X-S$, but note that we do not get exactly the same property for $G_1-X$ because of the still present sink-only copies of vertices in $X$. (The latter point will be no problem and should mainly explain why we need to talk about $G-X-S$ and not only $G_1$. Note that $G-S$ and $G_1$ by construction share the vertex set $V$ to be able to refer to connected components of $G-X-S$ and the graph $G_1$ underlying the gammoid $M_1$.)

While we do not know the paths in $\mathcal{P}_1$ entirely, we know for sure that no vertex of $X\cup\{x',x''\}$ can be an internal vertex of any path in $\mathcal{P}_1$ because there is a path ending in each of those vertices. Similarly, we may assume that no vertex of $T$ is internal to any path of $\mathcal{P}_1$: If not then any path $P\in\mathcal{P}_1$ with internal vertex from $T$ can be shortened to start in that vertex; this argument cannot be repeated indefinitely (as the paths get shorter each time). There is still at most one path intersecting the vertex set of any component of $G-X-S$. 

Now, define $T'\subseteq T$ as those vertices of $T$ in which no path of $\mathcal{P}_1$ starts; there must be exactly $|T|-|\mathcal{P}|=|T|-(|X|+2)$ of them since no vertex of $T$ is internal. Moreover, for each component $K$ of $G-X-S$, the set $T'$ contains all but at most one vertex of $T\cap V(K)$: At most one path of $\mathcal{P}_1$ can start in $T\cap V(K)$ and no vertex can be internal. This will be important for proving the claim below.

Clearly, the set $T'\cup X\cup\{x',x''\}$ is independent in $M_1$ because an appropriate path packing $\mathcal{P}'$ can be obtained from $\mathcal{P}_1$ by adding length zero paths for each $v\in T'$. The set $\hat X=\{\hat x \mid x\in X\}\subseteq \hat V$ is clearly independent in $M_2$ since it has size at most $k$. Thus, the set $I'=T'\cup X\cup\{x',x''\} \cup \hat X$ is independent in $M=M_1\oplus M_2$. Define $I$ as $I'\setminus\{x',x'',\hat x\}$, i.e., $I=T'\cup X \cup (\hat X \setminus \{\hat x\})$. The size of $I$ is at most
\[
|T'|+|X|+(|\hat X|-1)=|T|-(|X|+2) + |X| + |X| -1=|T|+|X|-3\leq |T|+k-3.
\]
Clearly, $\{x',x'',\hat x\}$ extends $I$, as $I'=\{x',x'',\hat x\} \cup I$ is independent and both are disjoint by choice of $I$. We now show that no other $\{v',v'',\hat v\}\in \mathcal{T}$ extends $I$.

\begin{claim}
If $\{v',v'',\hat v\}\in\mathcal{T}$ extends $I$ then $v=x$.
\end{claim}

\begin{proof}
Suppose that $\{v',v'',\hat v\}$ extends $I$. Clearly, this implies that $v\notin X\setminus \{x\}$ because otherwise $\{v',v'',\hat v\}$ would not be disjoint from $\hat X\setminus\{\hat x\}\subseteq I$. Thus, $v\in V\setminus (X\setminus\{x\})$.

Assume, for contradiction, that $v\in V\setminus X$, i.e., that $v\neq x$. We know that $\{v',v'',\hat v\}\cup I$ is independent in $M$, so $I_1:=I\cap V_1$ must be independent in $M_1$. Thus, there exists a collection $\mathcal{P''}$ of $|I_1|$ vertex-disjoint paths from $T$ to $I_1$ in $G_1$. Because $X\subseteq I_1$, the paths, say $P_{v'}$ and $P_{v''}$, from $T$ to $\{v',v''\}$ cannot have internal vertices from the set $X$. Furthermore, they cannot have other sink-only copies as internal vertices. Since $v\in V\setminus X$, this implies that $P_{v'}$ and $P_{v''}$ are entirely contained in some component $K_1$ of $G_1-(X\cup\{x',x''\mid x\in X\})$. (Component $K_1$ corresponds to a component $K$ of $G-X-S$ but also has sink-only copies of each vertex.) Recall now that in $T'$ we have all but at most one vertex of $T\cap V(K)$ for each connected component of $G-X-S$ and this is also true for $T\cap V(K_1)$ as $V(K_1)\cap V=V(K)$. Thus, in $\mathcal{P''}$ there is a path $v$ of length zero for each vertex $T'\cap V(K_1)$, leaving at most one vertex of $T$ to start paths to $\{v',v''\}$. This is a contradiction because $P_{v'}$ and $P_{v''}$ are entirely contained in $K_1$ and fully vertex-disjoint.

Thus, if $v\in V\setminus X$ then $\{v',v''\}\cup I_1$ is not independent in $M_1$ and, hence, $\{v',v'',\hat v\}$ does not extend $I$ in $M$. Together with the first paragraph this implies that $v=x$, as claimed.
\end{proof}

The set $I$ fulfills the required properties which completes the proof.
\end{proof}

We know now that for every vertex $x \in V \setminus T$ that is a vertex in a dominant solution the set 
$\{x',x'',\hat x \}$ is in every $(|T|+k-3)$-representative set $\mathcal{T'}$. If we define 
$V(\mathcal{T}') = \{ v \mid \{v',v'',\hat v\} \in \mathcal{T}'\}$ then this implies that $X_0 \subseteq V(\mathcal{T}')$ for each dominant solution $X$.
Thus, every dominant solution $X$ is contained in $V(\mathcal{T}') \cup T$.

\subparagraph{Shrinking the input graph to $\boldsymbol{\Oh(|V(\mathcal{T}') \cup T|)}$ vertices.}
In the previous parts we have shown that if there exists a solution for $(G,S,k)$, then there exists
a solution that is completely contained in $W:=V(\mathcal{T}') \cup T$. Using this we can make all 
vertices in $V \setminus W$ undeletable. We achieve this by applying the so-called $\torso$ operation to vertex set $W$ in $G$; let $G'= \torso(G,W)$. 
By definition of $\torso(G,W)$, the resulting graph $G'$ has vertex set $W$ and is derived from $G[W]$ by making each pair $\{u,v\}\subseteq W$ adjacent if there is a $u$,$v$-path in $G$ with internal vertices from $V\setminus W$. Note that we do not create double edges or loops in $G'$ and that all edges of $S$ are preserved in $G'$ because $T \subseteq W$.
(The same can be achieved by iteratively selecting a vertex $v\in V\setminus W$, making its neighbors a clique, and deleting $v$ from the graph.)

\begin{lemma} \label{lemma::equivalent_torso}
    $(G',S,k)$ has a solution if and only if $(G,S,k)$ has a solution.
\end{lemma}

It follows from Lemma \ref{lemma::equivalent_torso} that $(G',S,k)$ is an equivalent instance and the graph 
of this instance contains at most $|W|$ vertices. This completes the kernelization.
The correctness of Lemma \ref{lemma::equivalent_torso} follows from the fact that the $\torso$ operation 
preserves the separators that are contained in $W$ (cf.~\cite{DBLP:journals/corr/abs-1110-4765}). For 
completeness we give a short proof of the lemma.
\begin{proof}[Proof of Lemma \ref{lemma::equivalent_torso}]
Let $X$ be a solution for $(G',S,k)$. We prove that $X$ is also a solution for $(G,S,k)$ by contradiction.
Assume that $X$ is not a solution for $(G,S,k)$. Then there exists an $S$-cycle $C=v_1 v_2 \ldots v_l$ in 
$G-X$. 
Note that $S \subseteq E(G')$, because $T=V(S) \subseteq W$ and therefore at least two vertices of $C$ are contained in $W$.
Now we modify $C$ to obtain an $S$-cycle $C'$ in $G'$. Let $v_i,v_j \in W \cap C$ two vertices of the 
cycle with $i<j$ such that $\{v_{i+1}, \ldots, v_{j-1} \} \subseteq V \setminus W$. By definition there 
exists an edge $\{v_i,v_j\}$ in $\torso(G,W)$ and using these edges we obtain cycle $C'$. Note that $C'$ 
contains no vertex of $X$ and contains the same edges from $S$ that $C$ contains. Thus $C$ is an $S$-cycle in $G'-X$ which contradicts the assumption that $X$ is a solution of $(G',S,k)$.

For the other direction we assume that $(G,S,k)$ has a solution. Then there also exists a dominant 
solution $X$ for $(G,S,k)$ and we know that $X\subseteq W$. Again we prove that $X$ is also a solution for $(G',S,k)$ by contradiction.
Assume that $X$ is not a solution for $(G',S,k)$. Then there exists a path $P$ between the endpoints of 
an edge $e=\{x,y\} \in S$ in $G'-X$ that does not use the edge $e$. 
We modify $P'$ to obtain a path $P$ in $G$ that does not contain the edge $e$. 
If $P'$ uses an edge $\{u,v\}$ that is not contained in $G$, then there exists a $u$-$v$ path in 
$V \setminus W$ connecting $u$ and $v$. Crucially, $V\setminus W$ is disjoint from $X$ so this replacement still yields a walk that avoids $X$. Overall we get a walk from $x$ to $y$ in $G$ that does not contain
$e$ as an edge and that avoids $X$. This walk contains a path $P$ from $x$ to $y$ and this path together with the edge $e$
is an $S$-cycle in $G-X$ which is a contradiction to the assumption that $X$ is a solution for $(G,S,k)$.
\end{proof}

So far we have a kernelization that creates an equivalent instance $(G',S,k)$ such that $G'$ has $|W|$ vertices. As mentioned above, Lemma
\ref{lemma::representative} guarantees that $|W| \in \Oh(|S|^3)$ and this implies a polynomial
kernel for \ESFVS parameterized by $|S|$.
If we use the fact that the gammoid $M$ is the direct sum of two gammoids $M_1$ and $M_2$, and that all sets $\{v',v'',\hat v\}\in\mathcal{T}$ contain exactly two elements of $M_1$ and one element of $M_2$, then we can guarantee that $|W| \in \Oh(|S|^2 k)$, which is an improvement for all nontrivial instances with $k<|S|$.

\begin{lemma} \label{lemma::directed_sum_representative}
    Let $M=M_1 \oplus M_2$ be the gammoid of rank $|T|+k$ as defined above and $\mathcal{T}= \{I_1, I_2, \ldots, I_t\}$ be the set of independent sets of $M$ that we use for the kernelization.
    Let $A$ be represented by $\mathrm{diag}(A_1,A_2)$ as above.
    If $|\mathcal{T}| > \binom{|T|}{2} \cdot \binom{k}{1}$, then there exists a set $I \in \mathcal{T}$ such that $\mathcal{T} \setminus \{I\}$ is $(|T|+k-3)$-representative for $\mathcal{T}$.
\end{lemma}

The proof of Lemma \ref{lemma::directed_sum_representative} is similar to Marx~\cite[Lemma 4.2]{DBLP:journals/tcs/Marx09}. We additionally use the fact that $M$ is the direct sum of two gammoids to obtain that the vectors in the exterior algebra which represent the sets in $\mathcal{T}$ span a space of smaller dimension.
\begin{proof}[Proof of Lemma \ref{lemma::directed_sum_representative}.]
Let $U$ be the ground set of the matroid $M$ which equals the set of columns of $A$. For each $e \in U$, let $x_e$ be the corresponding $(|T|+k)$-dimensional column vector of $A$ and let $w_i = \bigwedge_{e \in A_i} x_e$ be a vector in the exterior algebra of the linear space $F^{|T|+k}$. Every $w_i$ is the wedge product of three vectors where exactly two are from $\bigl(\begin{smallmatrix} A_1 \\ 0 \end{smallmatrix} \bigr)$ and one from $\bigl(\begin{smallmatrix} 0 \\ A_2 \end{smallmatrix} \bigr)$. The two vectors corresponding to $\bigl(\begin{smallmatrix} A_1 \\ 0 \end{smallmatrix} \bigr)$ can only span a space of dimension $\binom{|T|}{2}$ and the vectors corresponding to $\bigl(\begin{smallmatrix} 0 \\ A_2 \end{smallmatrix} \bigr)$ can only span a space of dimension $\binom{k}{1}$.
Thus, the $w_i$'s span a space of dimension at most $ \binom{|T|}{2} \cdot \binom{k}{1} $. If $|\mathcal{T}| > \binom{|T|}{2} \cdot \binom{k}{1}$, then the $w_i$'s are not independent and there exists some vector $w_l$ that can be expressed as a linear combination of the other vectors. 

One can show analogously to Marx~\cite[Lemma 4.2]{DBLP:journals/tcs/Marx09} that $\mathcal{T} \setminus \{I_l\}$ is $(|T|+k-3)$-representative for $\mathcal{T}$. We replicate this proof for convenience of the reader. 
Assume that there exists a set $Y$ of size at most $|T|+k-3$ such that $I_l$ extends $Y$ and no other set $I_i$, $i \neq l$ extends $Y$. Let $y = \bigwedge_{e \in Y} x_e$. One property of the wedge product is that the product of some vectors in $F^{|T|+k}$ is zero if and only if they are not independent. Therefore it holds that $w_l \wedge y \neq 0$ and $w_i \wedge y = 0$ for every $i \neq l$. But $w_l$ is a linear combination of other $w_i$'s and by the multi-linearity of the wedge product we get that $w_l \wedge y \neq 0$ is a linear combination of the values $w_i \wedge y = 0$ for $i \neq l$, which is a contradiction.
\end{proof}

As mentioned above, Marx showed in~\cite{DBLP:journals/tcs/Marx09} that one can find in randomized polynomial-time a matrix with $r(M)$ rows that represents a given gammoid $M$. We can make this proof algorithmic in the same way Marx did ~\cite[Lemma 4.2]{DBLP:journals/tcs/Marx09}.
Combined with Lemma \ref{lemma::directed_sum_representative} it follows directly that we can find a $(|T|+k-3)$-representative subset $\mathcal{T'}$ of $|\mathcal{T}|$ whose size is at most $\binom{|T|}{2} \cdot \binom{k}{1} \in \Oh(|S|^2 k)$.
This leads to a polynomial kernel with $\Oh(|S|^2 k)$ vertices for \ESFVS parameterized by $|S|$ and $k$.

%%%%%%%%%%%%%%%%%%%%%%%%%%%%%%%%%%%%%%%%%%%%%%%%%%%%%%%%%%%%%%%%%%%%%%%%%%%%%%%%%%%%%%%%%%%%%%%%%%%%%%%%%%%%%%%%%%%%%%%%%%%%%%%%%%%%%%%%%%%%%%%%%%%%%%%%%%%%%%%%%%%%%%%%%%%%%%%%%%%%%%%%%%%%%%%%%%%%%%%%%%%%%%%%%%%%%%%%%%%%%%%%%%%%%%%%%%%%%%%%%%%%%%%%%%%%%%%%%%%%%%%%%%%%%%%%%%%%%%%%%%%%%%%%

\section{Reducing the size of \texorpdfstring{$\boldsymbol{S}$}{S}}\label{section:reducing}

We have seen that \ESFVS parameterized by $|S|$ and $k$ has a polynomial kernel.
Now the goal is to reduce the size of the set $S$ until $|S|$ is polynomially bounded in $k$. 
This will lead to a polynomial kernel of \ESFVS parameterized by $k$.

To begin, we do some initial modifications to ensure that we can always find a solution of size at most $k$ that contains no vertex of the set $V(S)$, if one exists. For this we first delete all vertices $v \in V$ with the property that $e=\{v,v\} \in S$ is a loop in $G$; since the vertex $v$ must be in any solution, we decrease the value $k$ by one. Next we delete all remaining loops, because these loops are not in $S$ and cannot be contained in any $S$-cycle. We also reduce the number of edges between two vertices $v,w \in V(G)$. If no edge that is incident to $v$ and $w$ is contained in the set $S$, then we delete all except one edge. On the other hand, if at least one edge between $v$ and $w$ is contained in $S$, then we delete all except two edges. One of these edges is contained in $S$ and the other not.
In the next step we add for every edge $e=\{v,w\} \in S$ two new vertices $v_e,u_e$ to the graph, subdivide the edge $e$ into three edges $\{v,v_e\},\{v_e,w_e\}, \{w_e,w\}$, and edit $S$ by replacing edge $e$ by the edge $\{v_e,w_e\}$ in $S$. If a solution $X$ of \ESFVS contains a vertex $x_e \in V(S)$, then we can instead add the vertex $x$ to $X$ and delete $x_e$ from $X$, because every cycle that contains vertex $x_e$ also contains vertex $x$; hence we can always find an optimal solution that is disjoint from $V(S)$.

Let $(G,S,k)$ be an instance of \ESFVS, such that $G$ is a graph with the above properties.
Analogous to the paper of Cygan et al.~\cite{DBLP:journals/siamdm/CyganPPW13} we consider a solution $Z$ of the \ESFVS, with the difference that our solution is an $8$-approximation of the problem, to reduce the size of $S$. 
Even et al.~\cite{DBLP:journals/siamcomp/EvenNZ00} show that there exists an $8$-approximation algorithm for \SFVS. Since \SFVS and \ESFVS are equivalent (cf.~\cite{DBLP:journals/siamdm/CyganPPW13}), we can compute in polynomial time an $8$-approximation for \ESFVS and we can assume that $Z \cap V(S) = \emptyset$.
If $|Z| > 8k$, then we can stop immediately because no solution of size at most $k$ can exist.
On the other hand, if $|Z| \leq k$, then $Z$ is a solution for the problem and we are done.

The set $Z$ is a feasible solution to \ESFVS on $(G,S,|Z|)$. This implies that every edge $e \in S$ is a bridge in $G-Z$. In a next step we also remove all edges in $S$ from $G-Z$. Every connected component in $G-Z-S$ contains no edge from $S$ and, following Cygan et al.~\cite{DBLP:journals/siamdm/CyganPPW13}, we call such a component a \emph{bubble}. We denote the set of bubbles by $\mathcal{D}_Z$ and define a graph $H_Z=(\mathcal{D}_Z, E_{\mathcal{D}_Z})$ whose vertices are bubbles and with bubbles $I$ and $J$ being adjacent, i.e., $\{I,J\} \in E_{\mathcal{D}_Z}$, if and only if the components $I$ and $J$ are connected by an edge from $S$. The graph $H_Z$ is a forest, because $Z$ is a solution for $(G,S,|Z|)$ and a cycle in $H_Z$ would give rise to an $S$-cycle in $G-Z$. Similarly, no two bubbles can be connected by more than one edge of $S$. By $V_I$ we denote the vertices that are contained in bubble $I$.
Since $|E(V_I,V_J) \cap S| \leq 1$ for all $I,J \in \mathcal{D}_Z$ and equality holds if and only if $\{I,J\} \in E_{\mathcal{D}_Z}$, we can associate an edge $e=\{I,J\} \in E_{\mathcal{D}_Z}$ with the one edge $e_S=\{v_I,v_J\}$ in $E(V_I,V_J) \cap S$.
If we add the vertex set $Z$ and all edges $\{z,I\}$ with the property that $z \in Z, I \in \mathcal{D}_Z$ and $E(z,V_I) \neq \emptyset$ to the graph $H_Z$ we obtain a graph $H_Z^+$ that contains $S$-cycles.
Note that every $S$-cycle must contain a vertex of the set $Z$. We partition the set of bubbles 
according to the number of bubbles they are connected with.
\begin{definition}
    A bubble $I \in \mathcal{D}_Z$ is called%
    \begin{inparaenum}[(i)]
        \item solitary, if $\deg_{H_Z}(I)=0$;
        \item leaf, if $\deg_{H_Z}(I)=1$; and 
        \item inner, if $\deg_{H_Z}(I)\geq 2$.
    \end{inparaenum}
    By $\mathcal{D}_Z^s, \mathcal{D}_Z^l ,\mathcal{D}_Z^i$ we denote the corresponding sets of of bubbles.
\end{definition}

Let $X \subseteq V \setminus V(S)$ be a superset of $Z$. We define the graphs $H_X$, $H_X^+$ as well as the sets $\mathcal{D}_X, E_{\mathcal{D}_X}$ analogously to the graphs $H_Z,H_Z^+$ and the sets $\mathcal{D}_Z, E_{\mathcal{D}_Z}$.
Observe that the number of edges in $S$ is at most $|\mathcal{D}_Z \setminus \mathcal{D}_Z^s|$, because 
$H_Z$ is a forest, any two bubbles are connected by at most one $S$-edge, and $V(S) \cap Z = \emptyset$.

So far our setup is essentially the same as the one used by Cygan et al.~\cite{DBLP:journals/siamdm/CyganPPW13}. However, instead of an $8$-approximate solution they use the framework of iterative compression, which provides a solution $Z$ of size $k+1$ and leaves them with the task of reducing the number of $S$-edges for the problem of finding a solution $Z^*$ that is \emph{disjoint} from $Z$. Moreover, it suffices for them to consider the case that every feasible solution (if one exists) is disjoint from $Z$. In this setting they are able to reduce to an equivalent instance (or find that some assumption was violated) with only $\Oh(k^3)$ edges in $S$.

Thus, while many relevant structures like $z$-flowers or parallel $x$-$y$ paths containing $S$-edges are the same, many things have to be handled differently. In particular, if we find that at least one out of two vertices $x,y\in Z$ must be in the solution then we cannot stop (using the maximality condition) but need to continue and use this information in a more direct way.

During the reduction we detect certain pairs $\{x,y\}$ of different vertices with the property that each solution of size at most $k$ must contain at least one of the vertices (if one exists).
We store this fact as a \emph{pair-constraint}. We keep and enforce this information in the final instance, unless we decide earlier to delete $x$ or $y$.
By $\mathcal{P}$ we denote the set of pair-constraints that we have found so far. We can interpret this 
set as a set of edges and by $V(\mathcal{P})$ we denote all vertices that are contained in a pair-constraint.
Note that vertices from the set $V(S)$ are never contained in a pair-constraint from $\mathcal{P}$, because there always exists a solution that is disjoint from $V(S)$.
We need the set $\mathcal{P}$ to detect edges in $S$ that may be safely deleted. To this end, we generalize the \ESFVS problem by adding a set of pair-constraints $\mathcal{P}$ to the input; we call this problem \PESFVS. 
\problembox{Pair-Constrained Edge Subset Feedback Vertex Set}{$k$}{
An undirected graph $G$, a set $S \subseteq E$ of edges, a set $\mathcal{P}$ of pair-constraints and an integer $k$.}{Does there exist a set $X \subseteq V$ of size at most $k$ such that $G-X$ contains no $S$-cycle and such that for each pair-constraint $\{x,y\} \in \mathcal{P}$ we have $x \in X$ or $y \in X$?}
Clearly, instances $(G,S,k)$ of \ESFVS and $(G,S,\emptyset,k)$ of \PESFVS are equivalent.
Our goal is to reduce the size of $S$ by detecting $S$-edges that we can delete from $S$ without changing the outcome. This leads to the following definition:
\begin{definition}
Let $(G,S,\mathcal{P},k)$ be an instance of \PESFVS. We call an edge $e \in S$ \emph{irrelevant}, if $X \subseteq V(G)$ is a solution for $(G,S,\mathcal{P},k)$ if and only if $X$ is a solution for $(G,S \setminus \{e\}, \mathcal{P},k)$.
\end{definition}

Note that if two different $S$-edges $e$ and $e'$ are irrelevant in $(G,S,\mathcal{P},k)$, then $e'$ is not necessarily irrelevant in $(G,S \setminus \{e\},\mathcal{P},k)$.
In addition we do not expect to find all irrelevant edges or pair-constraints.

\subparagraph{The reduction rules.}
We now present our reduction rules. Throughout we assume that always the lowest numbered applicable rule is applied first. Correctness and efficiency of the overall reduction process will be proved later.

Let $(G,S,\mathcal{P}=\emptyset,k)$ be an instance for \PESFVS and let $Z$ be an $8$-approximation of 
this problem with $k<|Z| \leq 8k$ that is disjoint from $V(S)$.
In the following the graphs $G-Z$, $G-Z-S$, $H_Z$, and $H_Z^+$ are always defined with respect to the current 
instance $(G,S,\mathcal{P},k)$ of \PESFVS. Note that $Z \subseteq V$ and we delete vertices from $Z$ if
we delete the corresponding vertex in $V$.

\begin{description}
    \item[Rule \namedlabel{rule::k<=0}{1}:] If $k<0$, or if $k=0$ and there exists an $S$-cycle, then reduce 
                $(G,S,\mathcal{P},k)$ to some trivial false instance,  
                i.e. $G':=(\{x\},\{ e=\{x,x\} \})$, $S':=\{e\}$, $\mathcal{P'} = \emptyset$ and $k':=0$.
    \item[Rule \namedlabel{rule::cc}{2}:] Delete all bridges and all connected components not containing 
                any edge from $S$.
    \item[Rule \namedlabel{rule::bridge}{3}:] If there exists an edge $e \in S$ such that $e$ is a bridge 
                in $\left(V,E \setminus ( S \setminus \{e\})\right)$, then reduce to $S'=S \setminus \{e\}$.
\end{description}

Rules \ref{rule::cc} and \ref{rule::bridge} ensure that each bubble $I \in \mathcal{D}_Z$ is adjacent to a 
vertex in $Z$ in the graph $H_Z^+$, i.e. for all $I \in \mathcal{D}_Z$ we have $E_{H_Z^+}(V_I,Z) \neq \emptyset$: Since Rule \ref{rule::cc} is not applicable every bubble $I \in \mathcal{D}_Z$ must be adjacent to a bubble $J \in \mathcal{D}_Z \setminus I$, or a vertex in $Z$; otherwise $G[V_I]$ would be a connected component of $G$ that does not contain any edge from $S$ ($V_I$ was deleted in Rule \ref{rule::cc}). From Rule \ref{rule::bridge} follows that a bubble $I \in \mathcal{D}_Z$ must be adjacent to a vertex in $Z$; otherwise the edge $e \in N(V_I) \cap S$ would be a bridge in $\left(V,E \setminus ( S \setminus \{e\})\right)$.

\begin{description}
    \item[Rule \namedlabel{rule::pair1}{4}:] If there exists a vertex $v$ in the set $V(\mathcal{P})$ 
                that is contained in at least $k+1$ pair-constraints of $\mathcal{P}$, then we reduce to $G'=G-v$ and $k'=k-1$.
    \item[Rule \namedlabel{rule::pair2}{5}:] If $|\mathcal{P}|> k^2$ (and Rule \ref{rule::pair1} is not 
                applicable), then reduce $(G,S,\mathcal{P},k)$ to some trivial false instance.
    \item[Rule \namedlabel{rule::flower}{6}:] If there exists a $z$-flower of order $k+1$ in $G$ for a  
                vertex $z \in Z$, then we reduce to $G':=G-z$ and $k':=k-1$.
\end{description}

For the next rules we need a maximal matching $M$ in $H_Z$ that covers all inner bubbles $\mathcal{D}_Z^i$ in $H_Z$. Note that two adjacent leaf bubbles $I_1,I_2$ are not adjacent to an inner bubble and form a $K_2$ in $H_Z$, hence the edge $\{I_1,I_2\} \in E_{\mathcal{D}_Z}$ is contained in every maximal matching in $H_Z$. We use this matching to detect pair-constraints in $Z$. To this end we introduce the following definition:
Let $e=\{I,J\}$ be an edge in the matching $M$. We say $e$ \emph{sees} the pair $\{x, y\}$ of different 
vertices $x,y \in Z$ respectively the vertex $x \in Z$, if $\{I,x\},\{J,y\} \in E(H_Z^+)$ or $\{I,y\},\{J,x\} \in E(H_{Z}^+)$ respectively $\{I,x\},\{J,x\} \in E(H_Z^+)$. 

\begin{description}
    \item[Rule \namedlabel{rule::addpair1}{7}:] If at least $(k+2)$ edges in $M$ see a pair $\{x,y\}$ of
                different vertices in $Z$, then we add $\{x,y\}$ to the set of pair-constraints $\mathcal{P}$.
    \item[Rule \namedlabel{rule::useless1}{8}:] If there exists an edge $e \in M$ such that $e$ sees no                            
                single vertex $z \in Z$ and for every pair $\{x,y\}$ seen by $e$ the pair $\{x,y\}$ is a pair-constraint in $\mathcal{P}$, then remove $e_S$ from $S$ and $e$ from $M$.
                (Recall: If $e=\{I,J\} \in M \subseteq E(H_Z)$, then $e_S$ is the unique edge in $E(V_I,V_J) \cap S$.)
\end{description}

The matching $M$ is always recomputed if, through application of rules, it does no longer cover every inner bubble or is maximal when testing whether Rules \ref{rule::addpair1} or \ref{rule::useless1} apply (i.e., if the preceding rules do not apply). If $M$ does cover all inner bubbles but neither Rule \ref{rule::addpair1} nor \ref{rule::useless1} apply then, as we will prove later, this implies $|M|\in\Oh(k^3)$ and, hence, that there are at most $2|M|\in\Oh(k^3)$ inner bubbles.

Let $L = \mathcal{D}_Z^l \setminus V(M)$ be the set of leaf bubbles that are not covered by $M$.
Because the matching covers at least all inner bubbles, we know that $|S| \leq 2|M|+|L|$.
Therefore we have to find a reduction rule that reduces the number of leaf bubbles in $L$. Every leaf bubble in $L$ is adjacent to an inner bubble in $H_Z$, because $M$ covers all leaf bubbles that are not adjacent to an inner bubble.
To bound the number of leaf bubbles in $L$ we define for each $z \in Z$ a graph $G_z$ with the help of the following two sets. The first one, $L_z = N_{H_Z^+}(z) \cap L$, is the set of all leaf bubbles $I$ that are adjacent to $z$ in $H_Z^+$.
The other $V_z^i = \{v \in V \mid \exists J \in N_{H_Z^+}(L_z) \colon v \in V_J \}$ consists of all 
vertices that are contained in an inner bubble that is adjacent to a leaf bubble in $L_z$.

\begin{fleqn}
\begin{align*}
    V(G_z) &= \{z\} \cup L_z \cup V_z^i \\
    E(G_z) &= E_{H_Z^+}(z,L_z) \cup \{ \{I,w\} \mid \exists I \in L_z, v \in V_I, w \in V_z^i \colon \{v,w \} \in S  \} \cup (E(G[V_z^i]) \setminus S)
\end{align*}
\end{fleqn}

\begin{figure}[t]
    \centering
    \begin{tikzpicture}[scale=0.75]
        \node[fill, circle, inner sep = 2pt] (z) at (6cm, 0cm) [label=below:$z$] {};
        \node[draw,circle] (1) at ( 1cm,2cm) {};
        \node[draw,circle] (2) at ( 3cm,2cm) {};
        \node[draw,circle] (3) at ( 5cm,2cm) {};
        \node[draw,circle] (4) at ( 7cm,2cm) {};
        \node[draw,circle] (5) at ( 9cm,2cm) {};
        \node[draw,circle] (6) at (11cm,2cm) {};
        \node[draw,ellipse, minimum width=2cm, minimum height=1cm] (a) at (2cm,5cm) {};
        \node[draw,ellipse, minimum width=2cm, minimum height=1cm] (b) at (5cm,5cm) {};
        \node[draw,ellipse, minimum width=2cm, minimum height=1cm] (c) at (9cm,5cm) {};
        \foreach \i/\n in {1/a,2/a,3/b,4/c,5/c,6/c}
            \draw[red, very thick, dashed] (\i) -- (\n);
        \foreach \i in {1,...,6}
            \draw (z) -- (\i);
        \node[draw,ellipse, dotted, minimum width=12cm, minimum height=1cm] at (6cm,2cm) [label=right:$\quad L_z$] {};
        \node[draw,ellipse, dotted, minimum width=12cm, minimum height=2cm] at (6cm,5cm) [label=right:$\quad V_z^i$] {};
    \end{tikzpicture}
         \caption{Graph $G_z$, \mytikzmark $S$-edges in $G_z$}
\end{figure}
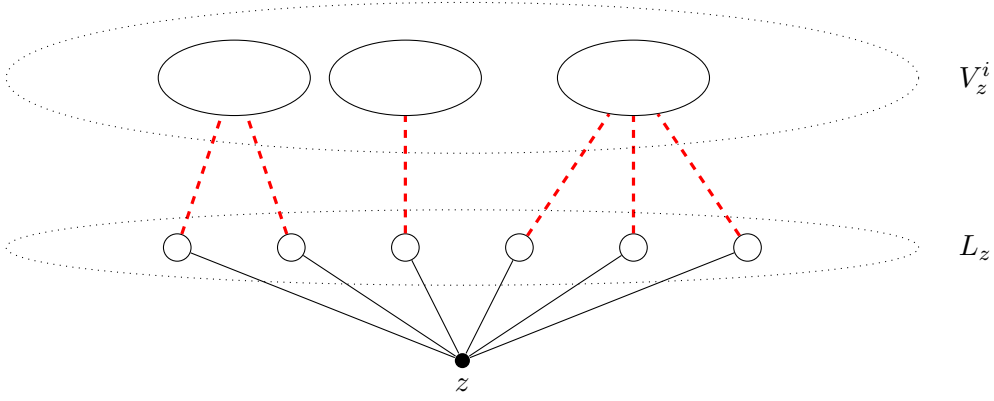

In the graph $G_z$ each leaf bubble $I \in L_z$ is a single vertex. We are not interested in the internal structure of leaf bubbles in $L_z$, whereas we are interested in the structure of the inner bubbles that are adjacent to the leaf bubbles in $L_z$. Thus we add the connected component that corresponds to an inner bubble which is adjacent to a bubble in $L_z$ to $G_z$.
In order to apply the concept of flowers and blocking sets in $G_z$, an edge $e \in E(G_z)$ is an $S$-edge in $G_z$ if $e=\{I,w\}$ with $I \in L_z$ and $w \in V_z^i$. Note that $e$ is an edge in $G_z$, because there exists an $S$-edge $e'=\{v,w\}$ in $G$ with $v \in V_I$.

\begin{lemma} \label{lemma::z-flower}
    If there exists no $z$-flower of order $k+1$ in $G_z$ for a vertex $z \in Z$, then we can find a $z$-blocker $B_z \subseteq V_z^i \setminus V(S)$ of size at most $2k$ in $G_z$.
\end{lemma}
The lemma follows from Theorem \ref{theorem::gallai} and the preprocessing as well as the 
construction of $G_z$.

\begin{proof}[Proof of Lemma \ref{lemma::z-flower}]
    The number of vertex-disjoint $L_z$-paths in $G_z -z$ is at most $k$, otherwise the $L_z$-paths together with vertex $z$ would correspond to a $z$-flower of order $k+1$ in $G_z$; this contradicts the assumption. From Theorem \ref{theorem::gallai} it follows that there exists a set $B_z \subseteq V(G_z-z)=L_z \cup V_z^i$ of size at most $2k$ intersecting every $L_z$-path. Since every $S$-cycle through $z$ in $G_z$ must contain an $L_z$-path, $B_z$ is a $z$-blocker of size at most $2k$ in $G_z$.
    
    It remains to show that there exists a $z$-blocker $B_z \subseteq V_z^i \setminus V(S)$.
    First we assume that there exists a vertex $I \in B_z \cap L_z$. From the construction of $G_z$ it follows that every leaf bubble $I \in L_z$ has degree one in $G_z-z$. Thus instead of $I$ we can choose the vertex in $N_{G_z}(I)$ for the $z$-blocker $B_z$ to obtain that $B_z \subseteq V_z^i$.
    
    In the next step we take care that $B_z$ is also disjoint from $V(S)$.
    Assume that $B_z$ contains a vertex $v_e \in V(S) \cap V_z^i$. From the preprocessing it follows that we can add $v \subseteq V_z^i \setminus V(S)$ to $B_z$ and delete $v_e$ from $B_z$, because every cycle that contains $v_e$ also contains $v$.
    
    Note that we delete at least as many vertices from $B_z$ as we add to $B_z$, hence $B_z$
    is still of size at most $2k$.
\end{proof}

Since no previous rule is applicable and a $z$-flower of order $k+1$ in $G_z$ gives rise to a $z$-flower of order $k+1$ in $G$, we find a $z$-blocker of size at most $2k$ for every vertex $z \in Z$.
Let $B = \bigcup_{z \in Z} B_z$ be the union of all $z$-blockers $B_z$ of size at most $2k$. 
Note that the set $L$ is the union of all sets $L_z$ with $z \in Z$, because every leaf bubble is adjacent to a vertex in $Z$ due to Rule \ref{rule::cc}, hence $L= \bigcup_{z \in Z} L_z$.

The following lemma provides three nice properties of the graph $H_{Z \cup B}=(\mathcal{D}_{Z \cup B}, E_{\mathcal{D}_{Z \cup B}})$ which helps us to bound the number of leaf bubbles in $L \subseteq \mathcal{D}_{Z}^l$. To memorize: The set $\mathcal{D}_{Z \cup B}$ is the set of bubbles in 
$G-(Z \cup B)-S$ and two bubbles $I,J$ are adjacent in $H_{Z \cup B}$ if and only if $E(V_I, V_J) \cap S \neq \emptyset$.

\begin{lemma} \label{lemma::properties_H_(ZuB)}
    The graph $H_{Z \cup B}$ has the following properties:
    \begin{enumerate}
        \item For each bubble $I \in \mathcal{D}_{Z \cup B}$ there exists a bubble $J \in \mathcal{D}_Z$, 
                            such that $V_I \subseteq V_J$. \label{lemma::all_bubbles}
        \item For each leaf bubble $J \in \mathcal{D}_Z$ there exists a leaf bubble 
                        $I \in \mathcal{D}_{Z \cup B}$, such that $V_I = V_J$. \label{lemma::leaf_bubbles}
        \item Let $I,J \in L$ and $K \in \mathcal{D}_{Z \cup B}^i$, such that 
                        $\{I,K\},\{J,K\} \in E_{\mathcal{D}_{Z \cup D}}$.
                        For all $z \in Z$ it holds that $z \notin N_G(V_I)$ or $z \notin N_G(V_J)$.
                        \label{lemma::bound_L}
    \end{enumerate}
\end{lemma}

\begin{proof}
Property \ref{lemma::all_bubbles} holds because the set $B$ only splits bubbles of $G-Z-S$ further (because we are now looking at deleting $Z\cup B$ from $G-S$) and does not merge any two 
bubbles. 
Property \ref{lemma::leaf_bubbles} follows from the fact that the set $B$ is disjoint from the set of leaf bubbles.
Next we show Property \ref{lemma::bound_L} by contradiction. We assume that some $z\in Z$ is in $N_G(V_I)$ and in $N_G(V_J)$. Then $I$ and $J$ are both vertices of the graph $G_z$ and hence both are contained in the set 
$L_z$. The consequence is that there exists an $L_z$ path from bubble $I$ over bubble $K$ to bubble $J$ in $H_{Z \cup B}$ which can be extended to a $L_z$ path in $G_z$ not containing any vertex in $B$; this contradicts the fact that $B_z \subseteq B$ blocks all $L_z$-paths in $G_z$.
\end{proof}

From Lemma \ref{lemma::properties_H_(ZuB)} it follows that $L \subseteq \mathcal{D}_{Z \cup B}^l$; thus we can use $H_{Z \cup B}$ to bound the number of leaf bubbles in $L$.
Let $\boldsymbol{I} =\{ J \in \mathcal{D}_{Z \cup B}^i \mid E(L,J) \neq \emptyset \}$ be the set of inner bubbles in $H_{Z \cup B}$ that are adjacent to a leaf bubble in $L$. 
Clearly the number of edges between $\boldsymbol{I}$ and $L$ in $H_{Z \cup B}$ equals the number $|L|$. Instead of again using a matching to reduce this number we consider more carefully the properties of these edges.
For this we define the property of seeing a pair in a slightly different way. 
Let $e=\{I,J\}$ be an edge with $I \in \boldsymbol{I}$ and $J \in L$. We say that $e=\{I,J\}$ with $I\in \boldsymbol{I}$ and $J\in L$ sees the pair $\{x,y\}$ of different vertices $x\in Z\cup B$ and $y\in Z$, if $\{I,x\},\{J,y\} \in E(H_{Z \cup B}^+)$. Observe that a bubble in $L$ is never adjacent to a vertex in $B$ in the graph $H_{Z \cup B}$, because $B \subseteq \bigcup_{z \in Z} V_z^i \setminus V(S)$.

\begin{description}
    \item[Rule \namedlabel{rule::addpair2}{9}:] If at least $(k+2)$ edges 
                $\{I_1,J_1\}, \{I_2,J_2\}, \ldots \{I_l,J_l\}$ with $l \geq k+2$, $I_i \in \boldsymbol{I}$ and $J_i \in L$ for $1 \leq i \leq l$ see a pair $\{x,y\}$ of different vertices, such that $x \in Z \cup B$ is adjacent to $I_i$, $y \in Z$ is adjacent to $J_i$ for all $i \in \{1,2,\ldots,l\}$, then we add $\{x,y\}$ to the set of pair-constraints $\mathcal{P}$.
\end{description}

At first sight Rule \ref{rule::addpair1} and \ref{rule::addpair2} may seem somewhat similar, but on closer inspection on can observe a decisive difference. In Rule \ref{rule::addpair2} we consider only edges between two disjoint sets of bubbles, whereas the edges in $M$ can be between two inner bubbles, an inner bubble and a leaf bubble, or between two leaf bubbles. For this reason we can require in Rule \ref{rule::addpair2} that all bubbles in $\boldsymbol{I}$ are adjacent to $x$ and all bubbles in $L$ are adjacent to $y$; this is not possible in Rule \ref{rule::addpair1}. We will see later that we need the definite assignment of the bubbles to the vertices in $Z \cup B$ by applying Rule \ref{rule::addpair2}. 

\begin{description}
    \item[Rule \namedlabel{rule::useless2}{10}:] If there exists an edge $e=\{I,J\}$ with 
                $I \in \boldsymbol{I}$ and $J \in L$ such that $e$ sees no single vertex $z \in B \cup Z$ and for every pair $\{x,y\}$ seen by $e$ the pair $\{x,y\}$ is a pair-constraint in $\mathcal{P}$, then remove $e_S$ from $S$, delete $J$ from $L$ and replace $I$ by $I \cup J$ in $\boldsymbol{I}$.
\end{description}

If we delete an edge $e=\{I,J\}$ from $S$ by applying Rule \ref{rule::useless2}, then the consequence is that
bubbles $I$ and $J$ are now merged into a single bubble. Anyhow, it is sufficient to continue with Rule \ref{rule::addpair2}, because $M$ is still a matching that covers all inner bubbles in the current graph $H_Z$ and $B$ still has the properties of Lemma \ref{lemma::properties_H_(ZuB)} with respect to the current graph $H_{Z \cup B}$.
That the edge set $M$ is still a matching in $H_Z$ holds because we never delete an edge in $M$ or an endpoint of an edge in $M$; we only merge an endpoint of an edge in $M$ with an unmatched leaf bubble in $L$. The first two properties of Lemma \ref{lemma::properties_H_(ZuB)} obviously hold with respect to the current graph $H_Z$. That Property \ref{lemma::bound_L} also holds follows from the fact that the leaf bubbles that are still in $L$ are the same as before and adjacent to the same inner bubbles as before.

\subparagraph{The reduction rules are safe.}
First we show that our reduction rules are safe, i.e. that there exists a solution for $(G,S,\mathcal{P},k)$ if and only if there exists a solution for $(G',S',\mathcal{P}',k')$.
Note that Rules \ref{rule::k<=0}, \ref{rule::cc}, and \ref{rule::flower} are obviously safe and Rule \ref{rule::bridge} is safe because for every $S$-cycle  through an edge $e \in S$ that is a bridge in 
$(V,E \setminus ( S \setminus \{e\}))$ there is another $S$-edge $e'$ on the cycle.
Let us consider the set $\mathcal{P}$ of pair-constraints to see that Rules \ref{rule::pair1} and \ref{rule::pair2} are safe. The set $\mathcal{P}$ naturally leads to the graph $P=(V(\mathcal{P}),\mathcal{P})$ and has the property that we have to pick at least on vertex of each pair-constraint for a solution for $(G,S,\mathcal{P},k)$. Hence any solution for $(G,S,\mathcal{P},k)$ must contain a vertex cover of $P$. 
Thus, Rules \ref{rule::pair1} and \ref{rule::pair2} are direct analogues of classical reduction rules for the \problem{Vertex Cover} problem, and hence safe.
To show that the other rules are safe, we first show a technical Lemma about a property of edges in $H_{Z \cup B}$.

\begin{lemma} \label{lemma::disjoint_path_L}
    If two different edges $\{I_1,J_1\}$ and $\{I_2,J_2\}$ in $H_{Z \cup B}$ with 
    $I_1,I_2 \in \boldsymbol{I}$, $J_1,J_2 \in L$ see a vertex $z\in Z$, respectively a pair $\{x,y\}$ with $x\in Z\cup B$ and $y\in Z$ such that $\{x,I_1\}, \{x,I_2\}, \{y,J_1\}, \{y,J_2\} \in E(H_{Z \cup B})$, then it holds that they are disjoint, i.e. that $I_1 \neq I_2$ and $J_1 \neq J_2$.
\end{lemma}

\begin{proof}
    We first assume that $I_1 = I_2$. This implies that $J_1$ and $J_2$ are leaf bubbles in $L$ which are adjacent to the same inner bubble $I=I_1=I_2$ in $H_{Z \cup B}$. For $J_1$ and $J_2$ it must hold that $z \in N_G(V_{I_i})$ respectively $y \in N_G(V_{I_i})$ for $i=1,2$. But this is a contradiction to Property \ref{lemma::bound_L} of Lemma \ref{lemma::properties_H_(ZuB)}.
    On the other hand if $J_1=J_2$, then $I_1=I_2$ because every leaf bubble in $L$ sees only one 
    other bubble.
\end{proof}

To show that Rules \ref{rule::addpair1} and \ref{rule::addpair2} are safe, we have to prove that we only add a pair $\{x,y\}$ of vertices to the set $\mathcal{P}$ of pair-constraints if either $x$ or $y$ must be in each solution of size at most $k$. The $(k+2)$ edges that see a pair $\{x,y\}$ are pairwise disjoint, because $M$ is a matching and Lemma \ref{lemma::disjoint_path_L} holds. 
Hence we have at least $(k+2)$ disjoint $x$-$y$ paths in $H_Z^+$ respectively $H_{Z \cup B}^+$ which we can 
extend to at least $(k+2)$ disjoint $x$-$y$ paths in $G$. This is the reason why at least one of $x$ and $y$ must be in any solution and it is safe to add $\{x,y\}$ to $\mathcal{P}$ as a pair-constraint. 

It remains to show that Rules \ref{rule::useless1} and \ref{rule::useless2} are safe. For this we prove that the edges that we delete in these rules are irrelevant. First we prove the following lemma.
\begin{lemma} \label{lemma::correctness_irrelevant}
    Let $Y \subseteq V \setminus V(S)$ be a superset of $Z$, hence $G-Y$ contains no $S$-cycle.
    If $e=\{I,J\} \in H_Y$ sees no single vertex $y \in Y$ and for every pair $\{x,y\}$ with $x,y \in Y$ seen by $e$ the pair $\{x,y\}$ is a pair-constraint in $\cP$, then $e_S=\{v_I,v_J\}$ is irrelevant for the instance $(G,S,\cP,k)$.
\end{lemma}

\begin{proof}
    Let $e=\{I,J\} \in H_Y$ be an edge with the properties of the lemma and $e_S$ the single edge in $E(V_I,V_J) \cap S$. To show that $e_S=\{v_I,v_J\}$ is irrelevant for instance $(G,S,\cP,k)$ we have to show that $X \subseteq V(G)$ is a solution for $(G,S,\cP,k)$ if and only if $X$ is a solution for $(G,S \setminus \{e_S\},\cP,k)$.
    Since every solution $X$ for $(G,S,\cP,k)$ is also a solution for $(G,S \setminus \{e_S\},\cP, k)$, we only have to show the other direction.
    
    Let $X$ be a solution for $(G,S \setminus \{e_S\},\cP, k)$. We assume that there exists an $S$-cycle $C$ in $G-X$. This $S$-cycle $C$ can only contain the $S$-edge $e_S$; otherwise would $C$ be an $(S \setminus \{e_S\})$-cycle which contradicts the fact that $X$ is a solution for $(G,S \setminus \{e_S\},\cP,k)$.
    
    \begin{claim} \label{claim::cycle_sees_pair}
        If an $S$-cycle $C$ in $G$ only contains the $S$-edge $e_S$, then there exists either a vertex $y \in Y$ such that $e$ sees the single vertex $y$ and $y$ is contained in cycle $C$ or two different vertices $x,y \in Y$ such that $e$ sees the pair $\{x,y\}$ and cycle $C$ contains $x$ and $y$.
    \end{claim}
    
    \begin{proof}
        Let $C$ be an $S$-cycle with the properties of the claim.
        Thus $C$ must exit bubble $I$ and bubble $J$ by edges that end in $Y$, because this is the only way to obtain a path from $v_I$ to $v_J$ that uses no edge from $S$.
        If these two edges share their endpoint $y$ in $Y$, then $e$ sees the single vertex $y$ and $y$ is contained in $C$. On the other hand if these two edges have different endpoints $x,y$ in $Y$, then $e$ sees the pair $\{x,y\}$ and the vertices $x,y$ are contained in $C$.
    \end{proof}
    
    Based on Claim \ref{claim::cycle_sees_pair}, it follows that edge $e=\{I,J\}$ must see a single vertex $y \in Y$ that is contained in $C$ or a pair $\{x,y\}$ with $x,y \in Y$ such that $x,y$ are contained in $C$. From the properties of edge $e$ follows that $e$ sees no single vertex and every pair $\{x,y\}$ that is seen by $e$ must be contained in a pair-constraint. Let $\{x,y\}$ be the pair that is seen by $e$ such that $x,y$ are vertices of cycle $C$ (Claim \ref{claim::cycle_sees_pair}). But at least one vertex of the pair $\{x,y\}$ must be in the solution $X$ for $(G,S \setminus \{e_S\},\mathcal{P},k)$, since $e$ sees only pairs that are contained in the set $\mathcal{P}$ of pair-constraints; hence $C$ is no cycle in $G-X$.
\end{proof}

From Lemma \ref{lemma::correctness_irrelevant} follows that we only delete an edge $e_S$ in Rule \ref{rule::useless1} and \ref{rule::useless2} when $e_S$ is irrelevant for instance $(G,S,\cP,k)$;
this holds because $Y=Z$ respectively $Y=Z \cup B$ is a superset of $Z$.

\subparagraph{Applying the Rules.}
First we show that if none of the rules can be applied, then the size of $S$ is bounded by $\Oh(k^4)$.
For this we prove two lemmas. One bounds the size of $M$ which helps us to bound the number of inner 
bubbles and the other bounds the number of leaf bubbles in $L$.

\begin{lemma}
    If the matching $M$ covers all inner bubbles in $H_Z$ and we cannot apply Rules \ref{rule::k<=0} through
    \ref{rule::useless1}, then the size of $M$ is at most $\Oh(k^3)$.
\end{lemma}

\begin{proof}
    Each edge in $M$ sees either a pair of vertices in $Z$ or a single vertex in $Z$. The number of pairs in $Z$ is at most $\binom{|Z|}{2} \leq |Z|^2$. Therefore the number of pairs in $Z$ that are not in the set $\mathcal{P}$ of pair-constraints is at most $|Z|^2$. Because we cannot apply Rule \ref{rule::addpair1}, at most $(k+1)$ edges in $M$ see any pair that is not in the set of pair-constraints. Thus at most $(k+1)|Z|^2$ edges of $M$ can see a pair of vertices in $Z$ that is not in $\mathcal{P}$.
    The number of edges in $M$ that see a single vertex in $Z$ is at most $k |Z|$; otherwise we can apply Rule \ref{rule::flower}, because at least one single vertex $z$ in $Z$ is seen by at least $k+1$ edges from $M$ and these edges together with $z$ are a $z$-flower of order $k+1$ in $H_Z^+$ which we can expand to a $z$-flower of order $k+1$ in $G$. Since we cannot apply Rules \ref{rule::flower}, \ref{rule::addpair1} or \ref{rule::useless1}, this leads to at most $(k+1)|Z|^2 + k |Z| \in \Oh(k^3)$ edges in $M$, because $|Z| \leq 8k$.
\end{proof}

From the lemma it follows that the number of inner bubbles in $H_Z$ is at most $2|M| \in \Oh(k^3)$.

\begin{lemma}
    If we cannot apply Rules \ref{rule::k<=0} through \ref{rule::useless2} then the size of $L$ is bounded by $\Oh(k^4)$.
\end{lemma}

\begin{proof}
    We claim that the number of edges between bubbles in $\boldsymbol{I}$ and bubbles in $L$ is at most
    $(k+1)|Z|(|B|+|Z|)+k|Z|$, if no rule is applicable. 
    This implies that there are at most $\Oh(k^4)$ leaf bubbles in $L$.
    
    Each edge between bubbles in $\boldsymbol{I}$ and bubbles in $L$ sees a pair $\{x,y\}$, such that $\{x,I\},\{y,J\} \in E(H_{Z \cup B})$ with $x \in Z \cup B$ is adjacent to $I$, $y\in Z$ is adjacent to $J$ or a vertex $z$ in $Z$; hence the number of pairs is at most $|Z| (|Z|+|B|)$. Rule \ref{rule::addpair2} adds $\{x,y\}$ to $\mathcal{P}$ if at least $(k+2)$ edges $\{I_1,J_1\},\{I_2,J_2\}, \ldots, \{I_l,J_l\}$ with $l \geq k+1$, $I_i \in \boldsymbol{I}$ and $J_i \in L$ for $1 \leq i \leq l$ see the pair $\{x,y\}$ such that $x \in Z \cup B$ is adjacent to $I_i$ and $y \in Z$ is adjacent to $J_i$ for $1 \leq i \leq l$. This bounds the number of edges between vertices in $\boldsymbol{I}$ and $L$ which see a pair, whose vertices are not a pair in the set $\mathcal{P}$ of pair-constraints, by $(k+1) |Z| (|Z|+|B|)$. 
    The number of edges between vertices in $\boldsymbol{I}$ and $L$ that see a certain vertex $z$ is at most $k$, otherwise the at least $k+1$ edges between $\boldsymbol{I}$ and $L$ that see vertex $z$ together with vertex $z$ form a $z$-flower of order $k+1$ in $H_{Z \cup B}^+$ because Lemma \ref{lemma::disjoint_path_L} ensures that the edges are disjoint. But then we can apply Rule \ref{rule::flower} and delete vertex $z$. Hence at most $k |Z|$ edges between vertices $\boldsymbol{I}$ and $L$ can see a vertex in $Z$.
    This leads to at most $(k+1)|Z|(|B|+|Z|)+k|Z|$ edges between vertices in $\boldsymbol{I}$ and $L$, because we cannot apply Rules \ref{rule::flower}, \ref{rule::addpair2} or \ref{rule::useless2}; this implies that $|L| \in \Oh(k^4)$, because $|Z| \leq 8k$ and $|B| \leq 2k|Z| \leq 16 k^2$.
\end{proof}

If we combine these two results, we know that $|\mathcal{D}_Z^i| + |\mathcal{D}_Z^l| \in \Oh(k^4)$.
As mentioned above this is an upper bound for the number of edges in $S$, because $H_Z$ is a forest, because there is at most one edge of $S$ between any two bubbles, and because $V(S) \cap Z =\emptyset$.

Finally we have to prove that we can perform the reduction in polynomial time. First we prove that each rule is applied a polynomial number of times and second that every single rule application can be performed in polynomial time.

\begin{lemma}
    Each reduction rule is applied at most a number of times that is polynomially bounded in the input size.
\end{lemma}

\begin{proof}
Note that we reduce in each rule, except Rules \ref{rule::addpair1} and \ref{rule::addpair2}, the size of at least one of the sets $V$, $E$, $S$, the value $k$ or decide that no solution of size at most $k$ exists. 
In Rules \ref{rule::addpair1} and \ref{rule::addpair2} we add pair constraints to $\mathcal{P}$, but if $\mathcal{P}$ contains more than $k^2$ pair constraints, we either find a vertex $z \in V(\mathcal{P})$ that we delete in Rule \ref{rule::pair1} and reduce $k$ by one or we decide in Rule \ref{rule::pair2} that no solution of size at most $k$ exists. This bounds the number of pair constraints that we add to $\mathcal{P}$ during the reduction by $k^3$ because we can decrease $k$ at most $k$ times. 
Thus, each rule is applied at most a number of times that is polynomial in the input size.
\end{proof}

Next we show that each single rule application can be performed in polynomial time.
It is obvious that we can apply Rules \ref{rule::k<=0} through \ref{rule::pair2} in polynomial time. 
The following lemma addresses Rule \ref{rule::flower} by solving a matroid parity problem on an appropriate gammoid.
\begin{lemma}
Let $G=(V,E)$, $z\in V$, and $S\subseteq E$. A $z$-flower of maximum order, i.e., a maximum number of $S$-cycles that intersect only in $z$, can be found in (deterministic) polynomial time.
\end{lemma}

\begin{proof}
For simplicity, we assume that there are no edges of $S$ incident with $z$. If this is not the case, then it can be checked that for each neighbor $v\in N(z)$ with $\{v,z\}\in S$ removing $\{v,z\}$ from $S$ and adding instead all other edges incident with $v$ to $S$ gives the desired result. Furthermore, we assume that no two edges of $S$ are incident with the same vertex of $G$; this can be achieved by appropriate subdivision operations, without changing the maximum order of $z$-flowers.

Let $\{C_1,\ldots,C_t\}$ be a $z$-flower of order $t$. Each $C_i$ gives rise to a path $P_i$ between two different neighbors $u$ and $v$ of $z$; all these paths are fully vertex-disjoint. By our above assumption, there are no $S$-edges incident with $z$, hence, each $P_i$ must contain two consecutive vertices, say $s_i$ and $t_i$, with $\{s_i,t_i\}\in S$. In this way, each path $P_i$ can be split into two paths, $P_{i,s}$ and $P_{i,t}$, from $N(v)$ to $\{s_i,t_i\}$; all these $2t$ paths are pairwise vertex-disjoint and do not contain the vertex $z$. Thus, from any $z$-flower of order $t$ we get $2t$ vertex-disjoint paths in $G-z$ from $N(z)$ to $T\subseteq V(S)$, i.e., endpoints of $S$-edges, such that $T$ can be partitioned into $t$ two-sets of vertices that are also edges in $S$. In the language of gammoids this means that $T$ is an independent set in the gammoid on graph $G-z$, with sources $N(z)$, and ground set $V(S)$.

Conversely, any independent set $T$ in the mentioned gammoid implies the existence of $|T|$ vertex-disjoint paths in $G-z$ from $N(z)$ to $T$. If, as above, $T$ can be partitioned into edges of $S$ then this gives rise to a $z$-flower of order $t=|T|/2$: Clearly, $|T|$ must be even to allow for the partition into sets of size two. Moreover, the paths are vertex-disjoint and, thus, two paths from $N(z)$ ending in $\{s_i,t_i\}\in S$ can be combined, using that $\{s_i,t_i\}$ must be an edge of $G$ into a single path, say $P_i$, from $N(z)$ to $N(z)$ that contains at least one edge of $S$. Note that, because $s_i$ and $t_i$ are ends of two paths in the packing they cannot occur in any other paths, so this combination still yields vertex-disjoint paths in $G-z$. Finally, adding the vertex $z$, the paths $P_1,\ldots,P_t$ can be combined into $t$ $S$-cycles that intersect only in $z$.

Thus, the task of finding a $z$-flower of maximum order reduces to that of solving a matroid parity problem on a gammoid: The underlying graph is $G-z$, the source set is $N_G(z)$, the ground set is $V(S)$, and the pairs are given by $S$. Recall that pairs in $S$ are vertex-disjoint. Using the algorithm due to Lov\'asz ~\cite{Lovasz80}, one may find a maximum independent set composed of pairs in $S$ in polynomial time, when provided with a matrix representation for the gammoid. A small caveat would be that one would need a randomized algorithm for finding said representation. Conveniently, specialized deterministic algorithms exist for subclasses of linear matroids; we can use a deterministic algorithm due to Tong et al.~\cite{TongLV1984} that solves the problem by reduction to weighted matching on graphs. (Note that given a maximum independent set $T$ composed of pairs, the cycles of the $z$-flower can be found by simple disjoint paths computation for $N(z)$ to $T$ in $G-z$.)
\end{proof} 

It remains to show that we can apply Rules \ref{rule::addpair1} through \ref{rule::useless2} in polynomial time.
\begin{lemma}
    We can apply Rule \ref{rule::addpair1} and \ref{rule::useless1} in polynomial time.
\end{lemma}

\begin{proof}
    First of all we store for each edge $e=\{I,J\} \in M$ all vertices $z \in Z$ seen by edge $e$ and all pairs $\{x,y\}$ with $x,y \in Z$ seen by edge $e$. For each edge we need at most $\Oh(|Z|^2)$ time; we only have to test for each vertex $z \in Z$ respectively each pair $\{x,y\}$ with $x,y \in Z$ whether $\{I,z\},\{J,z\} \in E(H_Z)$ respectively $\{I,x\},\{J,y\} \in E(H_Z)$ or $\{I,y\},\{J,x\} \in E(H_Z)$.
    Next we count how many edges see a pair $\{x,y\}$ with $x,y \in Z$ and denote this value by $c_{\{x,y\}}$. It takes at most $\Oh(|E| |Z|^2)$ time to compute all values; we only have to count for how many edges we store a certain pair.
    If a counter $c_{\{x,y\}}$ has value at least $k+2$, then we add the pair $\{x,y\}$ to the set $\cP$ of pair-constraints. We can check this for all counters in $\Oh(|Z|^2)$ time.
    The above computation corresponds to the computation we need for Rule \ref{rule::addpair1}.
    To apply Rule \ref{rule::useless1} we only have to look at all vertices and pairs that we stored for an edge $e \in M$. If we have stored no single vertex and only pairs that are pair-constraints in $\cP$, then $e$ fulfills the conditions of an edge that we delete in Rule \ref{rule::useless1}. To check this for one edge takes at most $\Oh(|Z|^2)$ time.
\end{proof}

We prove that we can apply Rule \ref{rule::addpair2} and \ref{rule::useless2} in polynomial time similar to how we prove that we can apply Rule \ref{rule::addpair1} and \ref{rule::useless1} in polynomial time. 
We only have to remember which endpoint is adjacent to which vertex in a pair.

\begin{lemma}
    We can apply Rule \ref{rule::addpair2} and \ref{rule::useless2} in polynomial time.
\end{lemma}

\begin{proof}
    First of all we store for each edge $e=\{I,J\}$ with $I \in \boldsymbol{I}$, $J \in L$ all vertices $z \in Z$ seen by edge $e$ and all pairs $(x,y)$ with $x \in Z \cup B$ adjacent to $I$, $y \in Z$ adjacent to $J$ such that $e$ sees the pair $\{x,y\}$. For each edge $e=\{I,J\}$ with $I \in \boldsymbol{I}$, $J \in L$ we need at most $\Oh(|Z \cup B||Z|)$ time; we only have to test for each vertex $z \in Z$ respectively each pair $(x,y)$ with $x \in Z \cup B$, $y \in Z$ whether $\{I,z\},\{J,z\} \in E(H_{Z \cup B})$ respectively $\{I,x\},\{J,y\} \in E(H_{Z \cup B})$.
    Next we count for how many edges we stored the pair $(x,y)$ with $x \in Z \cup B$, $y \in Z$ and denote this value by $c_{(x,y)}$. It takes at most $\Oh(|E| |Z \cup B| |Z|)$ time to compute all values; we only have to count for how many edges we store a certain pair.
    If a counter $c_{(x,y)}$ has value at least $k+2$, then we add the pair $\{x,y\}$ to the set $\cP$ of pair-constraints. We can check this for all counters in $\Oh(|Z \cup B||Z|)$ time.
    The above computation corresponds to the computation we need for Rule \ref{rule::addpair2}, because we only store the pair $(x,y)$ for an edge if the edge sees the pair $\{x,y\}$.
    To apply Rule \ref{rule::useless2} we only have to look at all vertices and pairs that we stored for an edge $e$ between bubbles in $\boldsymbol{I}$ and bubbles in $L$. If we have stored no single vertex and only pairs $(x,y)$ such that $\{x,y\}$ is a pair-constraints in $\cP$, then $e$ fulfills the conditions of an edge that we delete in Rule \ref{rule::useless2}. To check this for one edge takes at most $\Oh(|Z \cup B||Z|)$ time.
\end{proof}

Finally, we show that we can compute the matching $M$ and the set $B$ in polynomial time.

\begin{lemma}
    We can compute a maximal matching $M$ in $H_Z$ that covers all inner bubbles in polynomial time.
\end{lemma}

\begin{proof}
    We prove the lemma by giving a simple greedy algorithm for this problem. Let $T$ be a connected component in $H_Z$. Since $H_Z$ is a forest, $T$ is a tree; take $T$ to be rooted in an arbitrary vertex $r$.
    \begin{algorithm}[H]
    \caption{Matching}\label{alg::matching}
    \begin{algorithmic}[1]
      \Require{A tree $T$ with root $r$.}
      \Ensure{A maximal matching $M$ that covers all inner bubbles in $T$.}
      \State let $v$ be a child of $r$ in $T$
      \State $M\leftarrow \{\{r,v\}\}$
      \ForAll{children $w$ of $v$ and $r$ with $w \neq v$ }
        \If{$w$ is no leaf in $T$}
            \State let $T_w$ be the subtree of $T$ rooted at $w$
            \State $M\leftarrow M \cup$ \Call{Matching}{$T_w,w$}
        \EndIf
      \EndFor\\
      \Return $M$
    \end{algorithmic}
    \end{algorithm}
    It can be easily shown that $M$ is an maximal matching in $T$ that covers all inner bubbles, because $T$ is a tree. Since we can apply this algorithm to every connected component in $H_Z$ in polynomial time, the union of all matchings leads to the required maximal matching.
\end{proof}

It remains to show that we can find $B$ in polynomial time. From Schrijver's proof of Theorem 
\ref{theorem::gallai} and the proof of Lemma \ref{lemma::z-flower} it follows that we can 
find in polynomial time either a $z$-flower of order $(k+1)$ or a $z$-blocker of size at most $2k$ in $G_z$.
Since there exists no $z$-flower in $G_z$ when we compute $B$, we compute for every $z$ exactly once the set $B_z$ and since $B$ is simply the union of all $z$-blockers we can compute $B$ in polynomial time.

\subparagraph{Finding an equivalent instance for Edge Subset Feedback Vertex Set.}
Up to now we can only bound the number of edges in $S$ for the \PESFVS problem. As mentioned above the instance $(G,S,\mathcal{P}=\emptyset,k)$ for \PESFVS is equivalent to the instance $(G,S,k)$ of \ESFVS. Therefore we only have to show that we can find in polynomial time an instance of \ESFVS that is equivalent to the instance $(G,S,\mathcal{P},k)$ of \PESFVS and has at most $\Oh(k^4)$ $S$-edges. 
Let $\{x,y\} \in \cP$ be a pair-constraint. If there are two edges between $x$ and $y$ of which at least one is contained in $S$, then $x$ or $y$ must be in any solution, because $xy$ is an $S$-cycle.
For this reason, the instance $(G',S'=S \cup \mathcal{P},k)$ of \ESFVS is equivalent to the instance $(G,S,\mathcal{P},k)$ of \PESFVS, where $G'$ is created from $G$ by adding one edge $\{x,y\}$ between every two vertices $x$ and $y$ with $\{x,y\} \in \mathcal{P}$ when $\{x,y\} \notin E$ and by adding an edge $\{x,y\}$ between $x$ and $y$ that is also contained in $S'$; hence there are two edges between $x$ and $y$ with $\{x,y\} \in \cP$ in graph $G'$ and we add exactly one edge between $x$ and $y$ to $S'$. 
Because we cannot apply Rule \ref{rule::pair1} or \ref{rule::pair2} to $(G,S,\mathcal{P},k)$, we know that $|\mathcal{P}| \leq k^2$. This leads to a bound of $|S|+|\mathcal{P}| \in \Oh(k^4)$ edges in $S'$ for the \ESFVS problem after the reduction.

Finally, we combine the results of Section \ref{section:kernelization} and Section \ref{section:reducing} to obtain a polynomial kernel for \ESFVS parameterized by $k$. Let us first make some comments about the reduction of the size of $S$ and the kernelization:
For the reduction of the size of $S$ we use the fact that we can always find a solution that is disjoint from $T$. This only holds because we modified the graph accordingly.
But since this is a correct reduction it holds that an input instance $(G,S,k)$ of \ESFVS has a solution if and only if the output instance $(G',S',k')$ of the reduction in Section \ref{section:reducing} has a solution.
Thus it is no problem that we consider dominant solutions for the kernelization in Section \ref{section:kernelization} and that the kernelization only guarantees the preservation of dominant solutions. Every instance $(G',S',k')$ has a dominant solution of size at most $k'$ when a solution of size at most $k'$ exists; remember that $X$ is a dominant solution for $(G',S',k')$ if it has minimum size and contains a maximal number of vertices from $T'$ among solutions of minimum size. Hence if $(G',S',k')$ has a solution then it has a dominant solution $X$ and $X$ is a dominant solution for $(G',S',k')$ if and only if $X$ is a dominant solution for $(G'',S',k')$ the output instance of the kernelization in Section \ref{section:kernelization}.

Summarized, the reduction of the number of edges in $S$ to $\Oh(k^4)$ edges together with the kernelization to $\Oh(|S|^2k)$ vertices for \ESFVS parameterized by $|S|$ and $k$, results in a kernelized instance with $\Oh(k^9)$ vertices for \ESFVS parameterized by $k$.

%%%%%%%%%%%%%%%%%%%%%%%%%%%%%%%%%%%%%%%%%%%%%%%%%%%%%%%%%%%%%%%%%%%%%%%%%%%%%%%%%%%%%%%%%%%%%%%%%%%%%%%%%%%%%%%%%%%%%%%%%%%%%%%%%%%%%%%%%%%%%%%%%%%%%%%%%%%%%%%%%%%%%%%%%%%%%%%%%%%%%%%%%%%%%%%%%%%%%%%%%%%%%%%%%%%%%%%%%%%%%%%%%%%%%%%%%%%%%%%%%%%%%%%%%%%%%%%%%%%%%%%%%%%%%%%%%%%%%%%%%%%%%%%%

\section{Conclusions}\label{section:conclusion}

We have shown that the \SFVS problem has a randomized polynomial kernelization using the matroid-based tools of Kratsch and Wahlstr\"om~\cite{DBLP:conf/focs/KratschW12}, positively answering the question of Cygan et al.~\cite{DBLP:journals/siamdm/CyganPPW13}. As in previous work~\cite{DBLP:conf/focs/KratschW12} the error-probability can be made exponentially small without increasing the kernel size. Nevertheless, it would of course be very interesting whether the use of randomization and/or matroids can be avoided. Furthermore, there is quite a gap between $\Oh(k^9)$ vertices and a lower bound of \emph{size} $\Oh(k^{2-\varepsilon})$ that is inherited from \problem{Vertex Cover}~\cite{DellM14}, conditioned on non-collapse of the polynomial hierarchy.

Other open problems regarding existence of polynomial kernels, possibly amenable to the matroid tools, are \MWC and \problem{Directed Feedback Vertex Set} (\DFVS). There is also a directed version of \SFVS, called \problem{Directed Subset Feedback Vertex Set}, but it generalizes \DFVS, whose kernel status has remained open for quite some time now.

%%%%%%%%%%%%%%%%%%%%%%%%%%%%%%%%%%%%%%%%%%%%%%%%%%%%%%%%%%%%%%%%%%%%%%%%%%%%%%%%%%%%%%%%%%%%%%%%%%%%%%%%%%%%%%%%%%%%%%%%%%%%%%%%%%%%%%%%%%%%%%%%%%%%%%%%%%%%%%%%%%%%%%%%%%%%%%%%%%%%%%%%%%%%%%%%%%%%%%%%%%%%%%%%%%%%%%%%%%%%%%%%%%%%%%%%%%%%%%%%%%%%%%%%%%%%%%%%%%%%%%%%%%%%%%%%%%%%%%%%%%%%%%%%

\bibliography{lit}

%%%%%%%%%%%%%%%%%%%%%%%%%%%%%%%%%%%%%%%%%%%%%%%%%%%%%%%%%%%%%%%%%%%%%%%%%%%%%%%%%%%%%%%%%%%%%%%%%%%%%%%%%%%%%%%%%%%%%%%%%%%%%%%%%%%%%%%%%%%%%%%%%%%%%%%%%%%%%%%%%%%%%%%%%%%%%%%%%%%%%%%%%%%%%%%%%%%%%%%%%%%%%%%%%%%%%%%%%%%%%%%%%%%%%%%%%%%%%%%%%%%%%%%%%%%%%%%%%%%%%%%%%%%%%%%%%%%%%%%%%%%%%%%%

\end{document}